\documentclass[lettersize,journal]{IEEEtran}  

\IEEEoverridecommandlockouts                              

\usepackage{graphicx}
\usepackage{cite}
\usepackage{amsmath}
\usepackage{adjustbox}
\usepackage{array}
\usepackage{subfigure}
\usepackage{multicol}
\usepackage{float}
\usepackage{setspace}
\usepackage{amsthm}
\usepackage{amssymb}
\usepackage{upgreek}
\usepackage{xcolor}
\definecolor{db}{rgb}{0.0, 0.2, 0.7}

\newtheorem{thm}{Theorem}[section]

\newtheorem{prop}[thm]{Proposition}

\newtheorem{defn}{Definition}[section]

\newtheorem{rem}{Remark}

\title{\LARGE \bf
Sequencing-enabled Hierarchical Cooperative CAV On-ramp Merging Control with Enhanced Stability and Feasibility
}

\author{Sixu Li, Yang Zhou, \textit{Member, IEEE}, Xinyue Ye, Jiwan Jiang and Meng Wang, \textit{Member, IEEE}
\thanks{\textit{(Corresponding author: Yang Zhou)}}
\thanks{Sixu Li and Xinyue Ye are with the Department of Multidisciplinary Engineering, Texas A$\&$M University, College Station, TX 77843, USA (email: sixuli@tamu.edu; xinyue.ye@tamu.edu).}%
\thanks{Yang Zhou is with the Zachry Department of Civil $\&$ Environmental Engineering, Texas A$\&$M University, College Station, TX 77843, USA (email: yangzhou295@tamu.edu).}%

\thanks{Jiwan Jiang is with the Department of Civil and Environmental Engineering, University of Wisconsin-Madison, Madison, WI 53706, USA (email: jiwan.jiang@wisc.edu).}%
\thanks{Meng Wang is with the Chair of Traffic Process Automation, Technische Universität Dresden, 01062 Dresden, Germany (email: meng.wang@tu-dresden.de).}%
}

\begin{document}

\maketitle
\thispagestyle{empty}
\pagestyle{empty}

\begin{abstract}

This paper develops a sequencing-enabled hierarchical connected automated vehicle (CAV) cooperative on-ramp merging control framework. The proposed framework consists of a two-layer design: the upper-level control sequences the vehicles to harmonize the traffic density across mainline and on-ramp segments, simultaneously enhancing lower-level control efficiency through a mixed-integer linear programming formulation. Subsequently, the lower-level control, in turn, employs a longitudinal distributed model predictive control (MPC) supplemented by a virtual car-following (CF) concept to ensure three key aspects: asymptotic local stability, $l_2$ norm string stability, and safety. Proofs of asymptotic local stability and $l_2$ norm string stability are mathematically derived. Compared to other prevalent asymptotic local-stable MPC controllers, the proposed distributed MPC controller greatly expands the initial feasible set. Additionally, an auxiliary lateral control is developed to maintain lane-keeping and merging smoothness while accommodating ramp geometric curvature. To validate the proposed framework, multiple numerical experiments are conducted. Results indicate a notable outperformance of our upper-level controller against a distance-based sequencing method. Furthermore, the lower-level control effectively ensures smooth acceleration, safe merging with adequate spacing, adherence to proven longitudinal local and string stability, and rapid regulation of lateral deviations. 

\end{abstract}


\section{INTRODUCTION}

\IEEEPARstart{O}{n-ramp} merging, identified as a cause of traffic void and a trigger of traffic disturbances\cite{ahn2007freeway, ahn2010}, has become a focal point in the realm of transportation science and engineering, primarily due to its adverse effects on traffic efficiency, safety, and energy efficiency \cite{ahn2010}. Conventional approaches like ramp metering, which uses traffic signals at the on-ramp to control the vehicle inflow onto freeways; and variable speed limits, which uses dynamic signage to adjust the speed limit on a roadway in response to changing traffic and weather conditions, have traditionally concentrated on regulating traffic flow at a macroscopic level \cite{lu2014review,papageorgiou2002freeway}. However, these methods often overlook the intricate aspects of individual vehicle-level operations, hindering the improvement in safety, efficiency, vehicle-level local stability, and system-level string stability.

The advent of recent connected automated vehicle (CAV) control technologies, facilitated by advancements in vehicle automation and vehicle-to-vehicle (V2V) as well as vehicle-to-infrastructure (V2I) communication, has presented remarkable prospects for revolutionizing transportation systems \cite{wang2015cooperative}. As a logical extension of fundamental driving functions such as CAV car-following (CF) and lateral control, CAV on-ramp merging control possesses the potential to enhance overall transportation efficiency and safety \cite{li2023beyond,li2024enhancing}.

The existing body of research on CAV-enabled on-ramp merging control predominantly concentrates on lower-level longitudinal control, which involves generating the appropriate longitudinal control input to facilitate safe and efficient merging. Various control strategies have been employed in this regard, including fuzzy control \cite{milanes2010automated}, linear feedback and feedforward control \cite{chen2021connected}, and model predictive control (MPC) \cite{cao2015cooperative}. Nevertheless, these approaches largely ignore the system-level string stability, which is a vital property to be against the disturbances (e.g., leading vehicles acceleration) \cite{zhou2019distributed}. Exceptions can be found in one study which explored the string stability of a virtual platoon controlled by a linear feedback and feedforward controller, and furnished mathematical proof through frequency domain analysis \cite{chen2021connected}. Whereas, neglecting the physical constraints such as acceleration/deceleration boundaries as well safety constraints greatly hinders its real-world application.

 Among the various control strategies, MPC has gained significant attention in recent years. This optimization-based control approach capably integrates vehicle models into constraints, respects vehicles' physical limitations, and optimizes multiple performance criteria \cite{borrelli2017predictive}. Additionally, the short horizon characteristic of MPC keeps the computation time short, making it well-suited for a real-time vehicle control strategy for on-ramp merging. Although to the best of the author's knowledge, local and string stability have rarely been proved for MPC-based lower-level longitudinal on-ramp merging controllers, MPC-based CF control approaches hold potential for extension to on-ramp merging lower-level control. This is supported by the extensive research in local stability \cite{zhou2019distributed, li2020distributed,dunbar2011distributed} and string stability \cite{zhou2019distributed, dunbar2011distributed} for MPC-based CF controllers. However, to guarantee local stability, these approaches typically constrain all the terminal states of a CAV (e.g., deviation from equilibrium spacing, speed difference, and/or acceleration at the end of a prediction horizon) to either zero \cite{dunbar2011distributed}, lie within an invariant set \cite{zhou2019distributed}, or match the average of the CAV's neighbors' assumed terminal states with an offset \cite{li2020distributed}. Such terminal state constraints lead to a highly restricted initial feasible set (i.e., the set of all the initial states that make the optimization problem feasible) \cite{borrelli2017predictive}, and these studies typically assume initial feasibility. This assumption, however, falls short in on-ramp merging scenarios where both mainline and on-ramp CAVs must be considered, and CAVs often do not begin near the equilibrium spacing. Consequently, initial feasibility cannot be assumed for these approaches and can only be achieved with longer prediction horizons\cite{li2020distributed,zhou2019distributed}, posing significant challenges to the direct online application of these MPC-based CF control strategies in on-ramp merging scenarios.

It is also important to highlight that studies concentrating solely on lower-level control often largely ignore the potential sub-optimal system-level performances resulting from the merging sequence. This oversight has prompted recent studies to develop hierarchical control approaches \cite{ding2019rule,hu2021embedding,chen2020hierarchical}, which use an upper-level controller to determine vehicle merging sequences and a lower-level controller to generate longitudinal control inputs. Main approaches used for the upper-level controller include rule-based control \cite{ding2019rule,hu2021embedding} and optimization-based control \cite{chen2020hierarchical}. However,those methods generally fail to consider the interactions with the lower-level controller and base merging sequences on estimated or minimal arrival times of the vehicles\cite{ding2019rule,hu2021embedding}. Some exemptions can be found, for example \cite{chen2020hierarchical}, which  attempts to address this issue by incorporating a second-order car-following model in the optimization-based upper-level controller with a
cooperative merging mode. Nevertheless, the upper-level control over-simplifies the lower level control policy and renders potential suboptimality. Another approach is to design a holistic framework, which generates the merging sequence and trajectories simultaneously \cite{rios2016automated,mu2021event}. Although some efforts have been undertaken to reduce the computational time \cite{mu2021event}, these approaches cannot fully fill the real-time control computation requirement by the complexity of formulation (e.g., as nonlinear mixed-integer programming) and problem scales. Additionally, these approaches tend to overlook vehicular lateral dynamics.

This paper aims to address the identified gaps and extend to real-world two-dimensional scenarios. Specifically, this paper develops a sequencing-enabled hierarchical CAV cooperative on-ramp merging control to ensure system-level optimality while considering two-dimensional vehicle dynamics. Moreover, this framework ensures both local and string stability in the lower-level longitudinal controller, and improves initial feasibility. The upper-level controller uses mixed-integer linear programming for vehicle sequencing, optimizing traffic density and vehicle movement. This optimal sequence is relayed to each CAV to create a virtual CF platoon. Managed by a distributed, MPC-based longitudinal controller and a linearized MPC-based lateral controller for lane-keeping, the system maintains system-level stability, safety, and smoothness.

The remainder of this paper is structured as follows. Section \ref{sec2} offers a brief introduction to the on-ramp merging scenario. Sections \ref{sec3} presents the upper-level controller. The development of the lower-level longitudinal controller, accompanied by mathematical proofs of local and $l_2$ norm string stability as well as an analysis of initial feasibility, is elucidated in Section \ref{sec4}. Section \ref{sec5} introduces the lower-level lateral controller. Section \ref{sec6} showcases the results and discussion of the simulation experiments. Finally, conclusions are drawn in Section \ref{sec7}.

\section{PROBLEM FORMULATION\label{sec2}}
This section defines a cooperative freeway on-ramp merging scenario as illustrated in Fig. \ref{fig:scenario}. In this scenario, a pure CAV environment is assumed. A coordinator placed near the merging point can communicate with all CAVs within the merging control area. Based on the positions and speeds of each vehicle, the coordinator is capable of computing the merging sequence and relaying the information to each vehicle via V2I. As depicted by the curved dashed arrows in the figure, the computed sequence enables the mapping of both mainline and on-ramp CAVs onto the same virtual Z-axis. the CAV merging control problem can be treated as a virtual CF problem via V2V as suggested by \cite{chen2021connected}, with an auxiliary lateral control. The unidirectional V2V is depicted by the solid straight arrows. Note that, the sequencing only needs to be conducted once a new CAV enters the control area, whereas the virtual CF control can be computed continuously over time, which helps to ensure the computation efficiency.

\begin{figure}[h]
    \centering
    \setlength{\abovecaptionskip}{0pt}
    \includegraphics[width=0.5\textwidth]{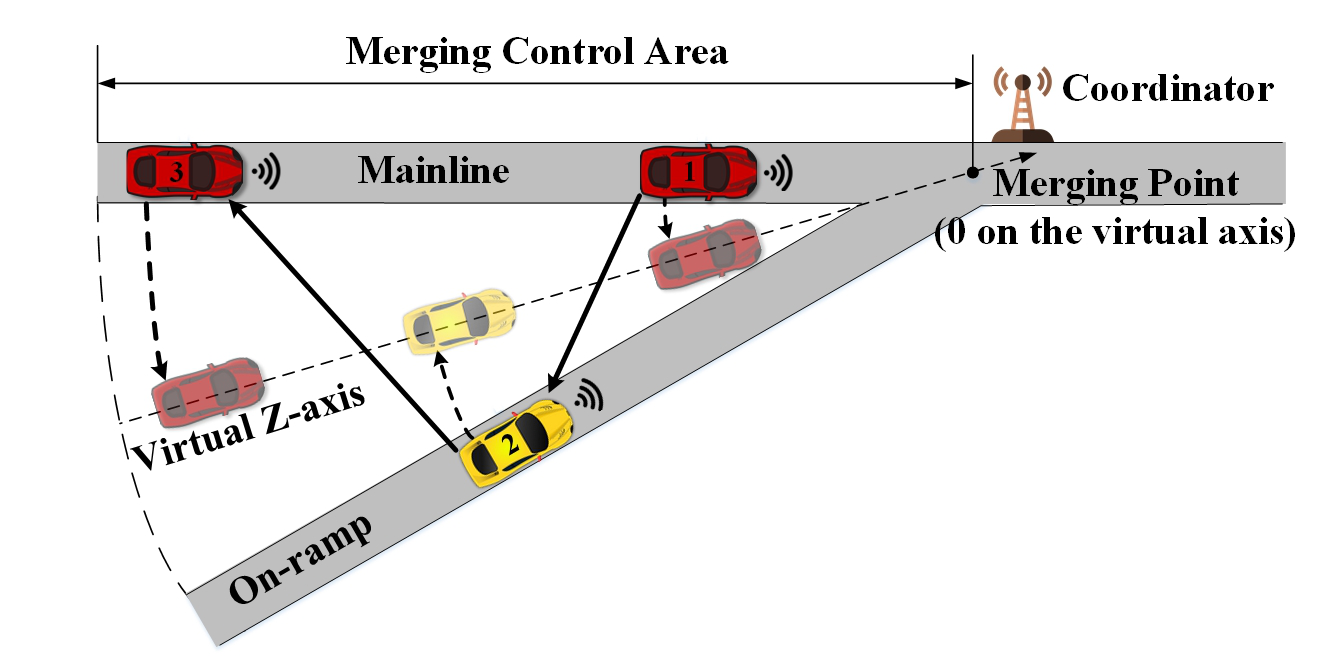}
    \caption{On-ramp merging scenario and communication}
    \label{fig:scenario}
\end{figure}

\begin{rem}
    While this paper focuses on scenarios with one lane on both the mainline and the ramp, our framework easily extends to scenarios with multiple lanes on the ramp and one lane on the mainline. In such scenarios, conflicts occur between all CAVs, hence they must be all mapped onto the virtual axis for effective virtual CF. Fig. \ref{fig:two lane scenario} demonstrates an example of the mapping with 2 lanes on the ramp, wherein MP1 and MP2 represent the merging points of lane 1 and lane 2, respectively. Scenarios with more lanes on the ramp follow similarly. As for scenarios with multiple lanes on the mainline, while our framework can be adapted by focusing solely on the right-most mainline lane, a more efficient method would involve allowing mainline CAVs to change lanes, thereby accommodating on-ramp CAVs. However, lane changing falls outside the scope of this paper, which may limit the efficiency of the proposed framework in such configurations.
\end{rem}
\begin{figure}[h]
    \centering
    \setlength{\abovecaptionskip}{0pt}
    \includegraphics[width=0.47\textwidth]{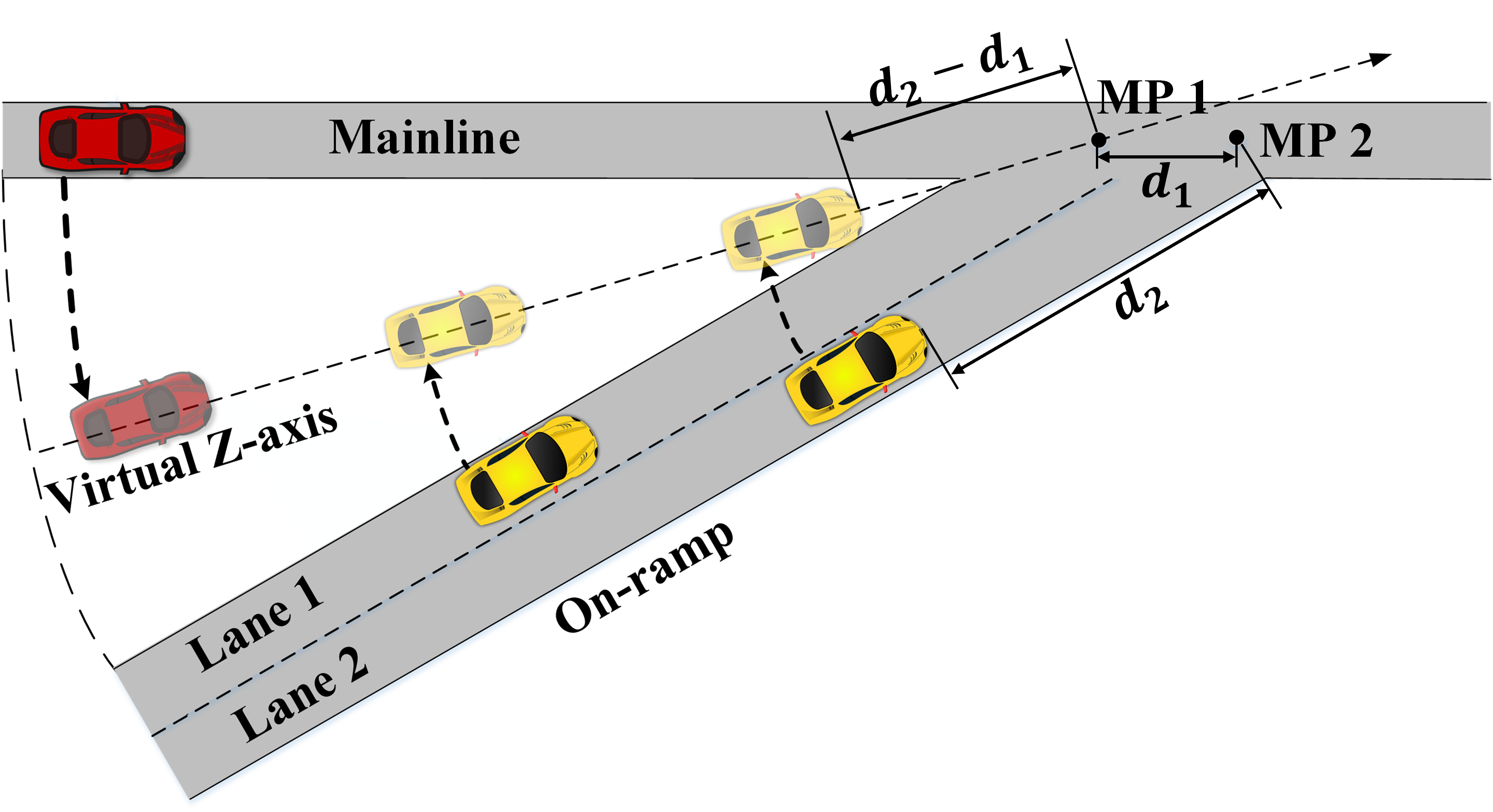}
    \caption{Virtual axis mapping in 2 ramp lanes scenario}
    \label{fig:two lane scenario}
\end{figure}

\section{UPPER-LEVEL CONTROLLER\label{sec3}}

 To determine a merging sequence that improves the efficiency and cost-effectiveness of the lower-level longitudinal controller while balancing traffic density on mainline and on-ramp, we designed an upper-level controller that optimizes multi-scale criteria related to both vehicle and traffic operation. The controller is formulated as a mixed-integer linear programming problem, enabling the realization of real-time solutions.  Subsections \ref{sec3.1} to \ref{sec3.3} introduce the decision matrix and information vectors, constraints, and cost function, respectively.

The position and velocity of CAVs are calculated on the virtual Z-axis as shown in Fig. \ref{fig:scenario}, with the origin placed at the merging point.
\subsection{Decision Matrix and Information Vectors\label{sec3.1}} 

We first introduce the decision matrix $U$ which is the output of the controller: 
\begin{equation}
    U=\begin{bmatrix}
        u_{1,1} & \cdots & u_{1,n} \\
        \vdots & \ddots & \vdots \\
        u_{m,1} & \cdots & u_{m,n} \\
        \vdots & \ddots & \vdots \\
        u_{n,1} & \cdots & u_{n,n}
    \end{bmatrix}
    \begin{matrix}
        \left.\begin{matrix}
            \\ \\ \\ 
        \end{matrix}\right\} m\text{ vehicles on mainline} \\
        \\
        \left.\begin{matrix}
            \\ \\
        \end{matrix}\right\} r\text{ vehicles on ramp~{}~{}~{}~{}~{}}
    \end{matrix}
    \label{eq:matrix}
\end{equation}

The matrix $U$ is a binary matrix, where the $n$ rows represent the $n$ CAVs on real roads, the rows are divided into two groups. The first group consists of the first $m$ rows, representing the $m$ CAVs on the mainline. The second group consists of rows $m+1$ to $n=m+r$, representing the $r$ CAVs on the ramp. Both groups are sorted by the CAV's distance from the merging point. The $n$ columns represent the $n$ virtual CF sequence IDs that will be assigned by the controller. The element $u_{i,j}$ takes on a value of 1 if the $i$th CAV on real roads is assigned with ID $j$ in the virtual CF sequence, and 0 otherwise. We define two constant information vectors that provide information about each CAV: the position vector ${P}=[{p}_1, {p}_2, \dots,  {p}_n]$, where $p_{i}$ is CAV $i$'s position on the virtual Z-axis, and the velocity vector ${V}=[{v}_1, {v}_2, \dots, {v}_n]$, with $v_{i}$ denoting CAV$i$'s longitudinal velocity. Both vectors are sorted as the rows of $U$. The coordinator gathers the position and velocity information of each CAV in the merging control area, enabling the controller to optimize the merging sequence based on the current conditions of the CAVs involved.

\subsection{Constraints\label{sec3.2} }
The constraints are divided into groups based on their meanings and purposes. First are the physical constraints:
\begin{equation}
\sum_{i=1}^{n} u_{i,j}=1,  ~{}j\in \{1,2,\ldots,n\} 
\label{eq:physic1}
\end{equation}\\[-3ex]
\begin{equation}
\sum_{j=1}^{n} u_{i,j}=1, ~{}i\in \{1,2,\ldots,n\} \label{eq:physic2}
\end{equation}\\[-3ex]
\begin{equation}
\setlength{\jot}{-6pt}
\begin{split}
u_{i,j} \leq 1-\sum_{l=j+1}^{n}u_{i-1,l}, ~{} i\in \{2,3,\ldots,m\}\\
,~{} j\in \{1,2,\ldots,n\} 
\label{eq:physic3}
\end{split}
\end{equation}\\[-3ex]
\begin{equation}
\setlength{\jot}{-6pt}
\begin{split}
u_{i,j}\leq 1-\sum_{l=j+1}^{n}u_{i-1,l}, 
~{} i\in \{m+2,m+3,\ldots,n\} \\
,~{} j\in \{1,2,\ldots,n\}
\label{eq:physic4}
\end{split}
\end{equation}

\noindent where (\ref{eq:physic1}) ensures that each virtual CF sequence ID can only be assigned to one CAV, while equation (\ref{eq:physic2}) ensures that each CAV can only be assigned with one virtual CF sequence ID. (\ref{eq:physic3}) restricts each CAV on the mainline not to exceed its predecessor on the same road, while (\ref{eq:physic4}) imposes the same restrictions for on-ramp CAVs, to prevent collision.

Secondly, we introduce constraints that restrict the value of continuous variable $\Delta d_{j}$ to be equal to the absolute value of CAV ($j+1$)'s deviation from the desired spacing, which will be minimized in the cost function:
\begin{equation}
\begin{split}
\Delta d_j \geq -\left(u_{:,j}^T P^T - u_{:,j+1}^T P^T - d_{j+1}^*\right)\\
,~{} j\in \{1,2,\ldots,n-1\}
\label{eq:d1}
\end{split}
\end{equation}
\begin{equation}
\Delta d_j \geq u_{:,j}^T P^T - u_{:,j+1}^T P^T - d_{j+1}^*,~{} j\in \{1,2,\ldots,n-1\} \label{eq:d2}
\end{equation}

\noindent where $d_{j+1}^*$ represents the constant desired spacing from the predecessor of the ($j+1$)th CAV in the virtual CF sequence. $u_{:,j}$ represents the $j$th column of matrix $U$, the term ($u_{:,j}^T P^T - u_{:,j+1}^T P^T - d_{j+1}^*$) represents the ($j+1$)th CAV's deviation from the desired spacing. To ensure that the continuous variable $\Delta d_{j}$ is greater than or equal to the absolute value of the spacing deviation of CAV $j+1$, we use (\ref{eq:d1}) and (\ref{eq:d2}). By further adding $\Delta d_{j}$ into the cost function linearly, as will be shown in subsection \ref{sec3.3}, $\Delta d_{j}$ is minimized and thus exactly equal to the absolute value of the spacing deviation of CAV $j+1$ \cite{zhang2016model}.

To account for the movement interactions between each CAV and its predecessor, we need to determine whether the spacing deviation is increasing or decreasing. We achieve this by obtaining two integer variables, $s_{d,j}$ and $s_{v,j}$, which represent the sign of the spacing deviation of CAV $j+1$ and the sign of the speed difference between CAV $j$ and $j+1$, respectively. Both variables are within $\left\{-1,1\right\}$. These variables are used to derive an indicator of the change in spacing deviation.

The constraints used to obtain $s_{d,j}$ are as follows:
\begin{equation}
u_{:,j}^T P^T - u_{:,j+1}^T P^T - d_{j+1}^* \leq M_1 y_{1,j},~{} j\in \{1,2,\ldots,n-1\} \label{eq:d4}
\end{equation}
\begin{equation}
\begin{split}
-(u_{:,j}^T P^T - u_{:,j+1}^T P^T - d_{j+1}^*) \leq M_2 y_{2,j}\\
,~{} j\in \{1,2,\ldots,n-1\} \label{eq:d5}
\end{split}
\end{equation}
\begin{equation}
y_{1,j} + y_{2,j} = 1,~{} j\in \{1,2,\ldots,n-1\} \label{eq:d6}
\end{equation}
\begin{equation}
s_{d,j} = y_{1,j} - y_{2,j},~{} j\in \{1,2,\ldots,n-1\} \label{eq:d7}
\end{equation}

\noindent The binary variables $y_{1,j}$ and $y_{2,j}$ are introduced with large constant numbers $M_{1}$ and $M_{2}$, respectively. The constraint (\ref{eq:d6}) requires that one of $y_{1,j}$ and $y_{2,j}$ is 0, and the other is 1. When the spacing deviation of CAV $j+1$ is negative, (\ref{eq:d4})-(\ref{eq:d6}) set $y_{1,j}=0$ and $y_{2,j}=1$.  Conversely, when the spacing deviation is positive, they set $y_{1,j}=1$ and $y_{2,j}=0$. Finally, the integer variable $s_{d,j}$ represents the sign of the spacing deviation of CAV $j+1$, with a value of -1 or 1, as a result of (\ref{eq:d7}).

The constraints used to obtain $s_{v,j}$ are similar:
\begin{equation}
u_{:,{j+1}}^T V^T - u_{:,{j}}^T V^T \leq M_3 y_{3,j}, ~{}j\in \{1,2,\ldots,n-1\} \label{eq:pos1}
\end{equation}
\begin{equation}
-(u_{:,{j+1}}^T V^T - u_{:,{j}}^T V^T) \leq M_4 y_{4,j}, ~{}j\in \{1,2,\ldots,n-1\} \label{eq:pos2}
\end{equation}
\begin{equation}
y_{3,j} + y_{4,j} = 1, ~{}j\in \{1,2,\ldots,n-1\} \label{eq:pos3}
\end{equation}
\begin{equation}
s_{v,j} = y_{3,j} - y_{4,j}, ~{}j\in \{1,2,\ldots,n-1\} \label{eq:pos4}
\end{equation}

\noindent where ($u_{:,{j+1}}^T V^T - u_{:,{j}}^T V^T$) is the speed difference between CAV $j$ and $j+1$. The approach used to obtain the binary variable $s_{v,j}\in\{-1,1\}$ follows the same technique and logic as in (\ref{eq:d4})-(\ref{eq:d7}),  so that $s_{v,j}$ eventually equals the sign of the speed difference between CAV $j$ and $j+1$.

Using $s_{d,j}$ and $s_{v,j}$, we can obtain integer variables $f_j$ that indicate whether the assigned sequences will reduce or increase the spacing deviations as (\ref{eq:f1}) and (\ref{eq:f2}):
\begin{equation} \label{eq:f1}
f_j \geq s_{d,j} - s_{v,j},~{} j\in \{1,2,\ldots,n-1\}
\end{equation}
\begin{equation} \label{eq:f2}
f_j \geq - (s_{d,j} - s_{v,j}),~{} j\in \{1,2,\ldots,n-1\}
\end{equation}

As can be found that if $s_{d,j} = s_{v,j}$, $f_{j}\geq 0$ and if $s_{d,j} \neq s_{v,j}$, $f_{j}\geq 2$. Minimizing $f_{j}$ helps to reduce the spacing deviation.

\subsection{Cost Function\label{sec3.3}}

As explained in subsection \ref{sec3.2}, we incorporate $\Delta d_{j}$ and $f_{j}$ in the cost function to transform inequality constraints into equality constraints. In fact, minimizing these variables is consistent with our microscopic vehicle operation criteria. Adding $\Delta d_{j}$ in the cost function minimizes the absolute value of spacing deviations. Similarly, adding $f_{j}$ in the cost function avoids forming predecessor-follower pairs with increasing spacing deviation. Additionally, we consider macroscopic traffic operation criteria by accounting for traffic density and prioritizing the road with higher traffic density for merging. The multi-scale cost function is formulated as:
\begin{equation} \label{eq:upper cost_function}
J^*(P,V,U) = \min_{u_{1,1} \rightarrow u_{n,n}} \underbrace{{ \sum_{j=1}^{n-1} (Q_{u} d_{j}+ R_{u}f_{j}) }}_{\text{Microscopic}}+\underbrace{{ \sum_{i=1}^{n} u_{i,:} S_{:,i} }}_{\text{Macroscopic}}
\end{equation}

\noindent where the first term improves microscopic vehicle movement by minimizing the absolute values of spacing deviation and the number of predecessor-follower pairs with increasing spacing deviation in the virtual CF sequence, $Q_u$ and $R_u$ are constant weights. The second term enhances macroscopic traffic operation by penalizing the CAVs on the lower-density road for merging first, where $u_{i,:}$ is the $i$th row of matrix $U$. To this end, we define a weight matrix $S$ that is designed before the mixed-integer programming is conducted and depends on the number of CAVs on both the mainline and ramp. Specifically, if CAV$i$ is on the road with lower traffic density, the $i$th column of $S$, represented by $S_{:,i}$, will be assigned a decreasing positive vector $\begin{bmatrix} 0.5^0,0.5^1 \cdots,0.5^{n-1} \end{bmatrix}^T$ to penalize the CAV for merging first. Otherwise, $S_{:,i}$ will be assigned with a zero vector.

\section{LOWER-LEVEL LONGITUDINAL CONTROLLER\label{sec4}}

To ensure smooth, stable, and safe merging while considering the physical limitations of vehicles, we developed a lower-level longitudinal controller using distributed MPC. The MPC controller minimizes the deviation from desired spacing, speed difference from the predecessor, and acceleration of each CAV. We underscore the  advantages and improvements of the proposed controller by mathematically proving its asymptotic local stability and $l_2$ norm string stability conditions, and analyzing its initial feasibility. Subsection \ref{sec4.1} introduces the platoon system dynamics, followed by subsection \ref{sec4.2}, which outlines the formulation of the optimization problem for the MPC. Proofs of asymptotic local stability and $l_2$ norm string stability are presented in subsections \ref{sec4.3} and \ref{sec4.4} respectively, while subsection \ref{sec4.5} provides an analysis on initial feasibility.

\subsection{Platoon System Dynamics \label{sec4.1}}

In this controller, we employ the constant distance policy \cite{dunbar2011distributed,li2020distributed}, which sets the desired spacing $d_i^*$ of CAV $i$ in the assigned virtual CF sequence as a predetermined constant value. Using this policy, the continuous platoon system dynamics can be expressed as follows:
\begin{equation} \label{eq:cd1}
\dot{\Delta d_i}(t)=\Delta v_i(t)
\end{equation}
\begin{equation} \label{eq:cd2}
\dot{\Delta v_i}(t)=a_{i-1}(t)-a_i(t)
\end{equation}
\begin{equation} \label{eq:cd3}
\dot{a_i}(t) = g_i(T_{des,i}(t), T_i(t), v_i(t))
\end{equation}

\noindent where $\Delta d_i(t) = d_i(t)-d_i^*(t)$ is the spacing deviation of CAV $i$ at time $t$, $d_i(t)$ is the current spacing of CAV $i$ at time $t$.
$\Delta v_i(t) = v_{i-1}(t)-v_{i}(t)$ is the speed difference between CAV $i-1$ and CAV $i$ at time $t$, $v_{i-1}(t)$ and $v_{i}(t)$ are the velocity of CAV $i-1$ and CAV $i$ at time $t$, respectively. $a_{i-1}(t)$ and $a_i(t)$ represent the acceleration of CAV $i-1$ and CAV $i$ at time $t$, respectively. $T_i$ and $T_{des,i}$ are the actual and desired driving/braking torques, respectively; $g_i(\cdot)$ is the vehicle driveline dynamics  \cite{xu2021energy,li2023sequencing}, given as:
\begin{equation} \label{eq:cd4}
\begin{split}
g_i(T_{des,i}(t), T_i(t), v_i(t))= \frac{\eta_i}{m_i r_i \tau_i} (T_{des,i}(t) - T_i(t)) \\- f_{roll,i}m_ig - \frac{1}{2} \rho C_d A_{v,i}v_i^2(t)
\end{split}
\end{equation}

\noindent where $\eta_i$ is the mechanical efficiency of the driveline, $m_i$ is the vehicle mass, $r_i$ is the tire radius, $\tau_i$ is the time lag of the driveline, $f_{roll,i}$ is the rolling resistance, $g$ is the gravitational acceleration, $\rho$ is the ambient air density, $C_d$ is the air drag coefficient, and $A_{v,i}$ is the vehicle's front projection area.

With the continuous state $x_i(t)= \begin{bmatrix} \Delta d_i(t), \Delta v_i(t), a_i(t) \end{bmatrix}^T$ defined, we can use the Euler discretization technique to obtain the discrete platoon system dynamics model as follows:
\begin{equation} \label{eq:ddm}
x_{i,k+1} =  Ax_{i,k} +  B\gamma_{i,k} + Da_{i-1,k} = f(x_{i,k},\gamma_{i,k},a_{i-1,k})
\end{equation}

\noindent where $A = \begin{bmatrix} 1 & T_s & 0 \\
0 & 1 & -T_s \\
0 & 0 & 1 \end{bmatrix}$, $B=\begin{bmatrix} 0 \\ 0 \\ T_s \end{bmatrix}$, $D=\begin{bmatrix} 0 \\ T_s \\ 0 \end{bmatrix} $. $T_s$ is the time interval between each time step.
$x_{i,k}=\begin{bmatrix} \Delta d_{i,k}, \Delta v_{i,k}, a_{i,k} \end{bmatrix}^T$ is the discrete state of CAV $i$ at time step $k$. $\gamma_{i,k}$ is the control input of CAV $i$ at time step $k$, defined as $\gamma_{i,k}=g_i(T_{des,i,k}, T_{i,k}, v_{i,k})$. The desired driving torque $T_{des,i}$ can be calculated by $T_{des,i}=g_i^{-1}(T_{i,k}, v_{i,k},\gamma_{i,k})$ and sent to the actuator.

\subsection{Optimization Formulation \label{sec4.2}}

For the lower-level longitudinal controller, a serial distributed MPC \cite{zhou2019distributed} is formulated, where MPC controllers solve the optimization problem sequentially in a CAV string, with each CAV's optimal control sequence solved based on the predicted states of its predecessor. This scheme improves the performance of each MPC controller in the string \cite{zhou2019distributed}. At each time step $t$, each CAV receives its predecessor's predicted acceleration sequence $\begin{bmatrix} a_{i-1,t}^*,\cdots,a_{i-1,t+N_p}^*\end{bmatrix}$, predicted position sequence $\begin{bmatrix} p_{i-1,t}^*,\cdots,p_{i-1,t+N_p}^*\end{bmatrix}$, and predicted velocity sequence $\begin{bmatrix} v_{i-1,t}^*,\cdots,v_{i-1,t+N_p}^*\end{bmatrix}$ via V2V communication as described in section \ref{sec2}. The predicted acceleration sequence is used for state update in the prediction horizon using (\ref{eq:ddm}). the predicted position sequence is used to determine if the safety cost becomes active, while the current position $p_{i-1,t}$ is also used to calculate the current spacing deviation $\Delta d_{i}(t)$. Finally, the predicted velocity sequence is used to calculate the constraint on speed difference while the current velocity $v_{i-1,t}$ is also used to calculate the current speed difference $\Delta v_{i}(t)$.

The MPC controller solves the following optimization problem at each time step:

\begin{equation} \label{eq:lowercost}
\begin{split}
J_i^*(x_{i,0},\Gamma_i)=\min_{\gamma_{i,0}\rightarrow\gamma_{i,N_p}}\sum_{k=0}^{N_p-1} l(x_{i,k},\gamma_{i,k})\\
+\beta l(x_{i,N_p},\gamma_{i,N_p})
\end{split}
\end{equation}

\noindent subject to:
\begin{equation} \label{eq:lowerconst1}
x_{i,0} = x_i(t)
\end{equation}\\[-3ex]
\begin{equation} \label{eq:lowerconst2}
x_{i,k+1} =  f(x_{i,k},\gamma_{i,k},a_{i-1,k}), ~{}k\in \{0,1,\cdots,N_p-1\}
\end{equation}\\[-3ex]
\begin{equation} \label{eq:lowerconst3}
\Delta d_{min} \leq \Delta d_{i,k} \leq \Delta d_{max}, ~{}k\in \{0,1,\cdots,N_p\}
\end{equation}\\[-3ex]
\begin{equation} \label{eq:lowerconst4}
v_{i-1,k}^*-v_{max} \leq \Delta v_{i,k} \leq v_{i-1,k}^*-v_{min}, ~{}k\in \{0,1,\cdots,N_p\}
\end{equation}\\[-3ex]
\begin{equation} \label{eq:lowerconst5}
a_{min} \leq a_{i,k} \leq a_{max}, ~{}k\in \{0,1,\cdots,N_p\}
\end{equation}\\[-3ex]
\begin{equation} \label{eq:lowerconst6}
\gamma_{min} \leq \gamma_{i,k} \leq \gamma_{max}, ~{}k\in \{0,1,\cdots,N_p\}
\end{equation}\\[-3ex]
\begin{equation} \label{eq:lowerconst8}
\Delta v_{i,N_p} = 0
\end{equation}\\[-3ex]
\begin{equation} \label{eq:lowerconst9}
a_{i,N_p} = a_{i-1,N_p}^*
\end{equation}\\[-3ex]

\noindent where (\ref{eq:lowercost}) is the cost function, the stage cost $l(x_{i,k},\gamma_{i,k})$ is defined as :
\begin{equation}\label{stage cost}
\begin{split}
l(x_{i,k},\gamma_{i,k})&=R_{lon} \gamma_{i,k}^2+ x_{i,k}^\top Q_{lon} x_{i,k}\\
&+P_{lon}\cdot e^{\frac{\Delta d_{i,0}}{-\Delta d_{safe}}}\Delta v_{i,k}^2\Theta(v_{i,0}, \Delta d_{i,0},k^*)
\end{split}
\end{equation}

\noindent with
\begin{equation} \label{safety cost conditions}
\Theta(\cdot)=\vartheta(-\Delta v_{i,0})\vartheta(-\Delta d_{i,0}-\Delta d_{safe})\vartheta(N_p-k^*)
\end{equation}

\noindent wherein $Q_{lon}$ is a diagonal positive weight matrix, $R_{lon}$ and $P_{lon}$ are constant weights, $\Delta d_{safe}$ is a predefined positive distance threshold, $k^*$ is a precalculated time threshold, and $\vartheta(\cdot)$ is the Heaviside function\cite{wang2014rolling,wang2015cooperative}. 

The cost function (\ref{eq:lowercost}) comprises two terms, the first term penalizes step zero to step $N_p-1$, while the second term penalizes step $N_p$, $\beta$ is a constant magnification and a lower bound of $\beta$ will be mathematically derived to ensure asymptotic local stability. The stage cost (\ref{stage cost}) comprises a state cost, a control input cost, and a safety cost that only activates when $\Delta d_{i,0}\leq-\Delta d_{safe}$, $\Delta v_{i,0}\leq0$, and $k^*\leq N_p$, the exponential term $e^{\frac{\Delta d_{i,0}}{-\Delta d_{safe}}}$ ensures that the CAV gets a large enough penalty when it gets too close to its predecessor (i.e., $\Delta d_{i,0}\leq-\Delta d_{safe}$) so that the relative speed decreases when the spacing decreases. $k^*$ is precalculated to ensure that the CAV will not collide with its predecessor. When the CAV and its predecessor are on different roads, $k^*$ can be calculated using a time-space diagram, as shown in Fig. \ref{fig:k*}, at each time $t$, within the prediction horizon, the assumed trajectory of CAV$i$ at each predicted step is calculated by $p_{i,k}^a=p_{i-1,k}^*-d_{i,0}(t)$, $k=0,1,\cdots,N_p$. $k^*$ is the first step satisfying $p_{i,k}^a\geq0$, if no such $k^*$ exists within the prediction horizon, then $k^*=\infty$. If the predecessor-follower are on the same road, $k^{*}=0$ for every step the optimization problem is executed.

\begin{figure}[h]
    \centering
    \setlength{\abovecaptionskip}{0pt}
    \includegraphics[width=0.4\textwidth]{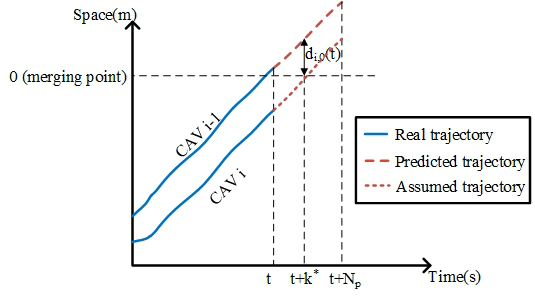}
    \caption{Calculation of $k^*$ using time-space diagram when CAV$i$ and CAV$i-1$ are on different roads}
    \label{fig:k*}
\end{figure}

The initial constraint (\ref{eq:lowerconst1}) restricts the optimization problem's initial state $x_{i,0}$ equal to the CAV's current state $x_i(t)$. (\ref{eq:lowerconst2}) is the discrete platoon system dynamics model in subsection \ref{sec4.1}.  (\ref{eq:lowerconst3}) sets boundary of the spacing deviation, it is important to note that (\ref{eq:lowerconst3}) is only incorporated to ensure the theoretical boundedness of the spacing deviation for the proof of local stability. Therefore, in practice, $\Delta d_{min}$ and $\Delta d_{max}$ can be any finite number as long as $\Delta d_{min}<\Delta d_{max}$. (\ref{eq:lowerconst4}) sets constraints on the speed difference where the lower and upper bound is calculated based on the speed limit, $v_{min}$ and $v_{max}$, and the predecessor's optimal speed, $v_{i-1,k}^*$. Therefore, (\ref{eq:lowerconst4}) is essentially the same as bounding the speed of CAV$i$ within $[v_{min},v_{max}]$. (\ref{eq:lowerconst5}) sets boundary on the acceleration while (\ref{eq:lowerconst6}) imposes constraints on the control inputs. (\ref{eq:lowerconst8}) and(\ref{eq:lowerconst9}) are the terminal constraints to ensure asymptotic local stability, furthermore, they ensure safety together with the safety cost, as they restrict that each CAV can adjust its speed and acceleration to match its predecessor's in $N_p$ steps.

\subsection{Local Stability Analysis} \label{sec4.3}

This subsection presents the asymptotic local stability (disturbance dissipation over time of each CAV) analysis of the controller proposed in subsection \ref{sec4.2}. Specifically, mathematical definition, conditions, and proof of asymptotic local stability is provided.

We consider the definition of Lyapunov local stability and asymptotic local stability adopted from\cite{willems1997introduction}:

\begin{defn}
For the system of $x_{i,k+1}=f(x_{i,k})$, the equilibrium point, $x_{i,e}$ is said to be Lyapunov locally stable if for every $\varepsilon$, there exists a $\delta>0$ such that if $||x_{i,0}-x_{i,e}||<\delta$, then for every $k \geq 0$, we have $||x_{i,k}-x_{i,e}||<\varepsilon$.
\end{defn}

\begin{defn}
For the system of $x_{i,k+1}=f(x_{i,k})$, the equilibrium point, $x_{i,e}$ is said to be asymptotically locally stable if it is Lyapunov locally stable, and there exists a $\delta>0$ such that if $||x_{i,0}-x_{i,e}||\rightarrow0$ as $t\rightarrow\infty$ if $||x_{i,0}-x_{i,e}||<\varepsilon$.
\end{defn}

The asymptotic local stability states that when a disturbance (i.e., deviation of state from the equilibrium state) is small enough (less than $\delta$), the disturbance will be completely dissipated over time, and the CAV is able to restore the equilibrium state.

Under the assumption that the leading vehicle CAV1 runs at a constant speed and the safety cost is inactive, we first prove the condition for the first follower in the virtual platoon, CAV2, to be asymptotically locally stable. Note that when CAV1 runs at a constant speed, $a_{1,k}=0$ for $k\in\{0,1,\dots,N_p\}$, therefore constraint (\ref{eq:lowerconst9}) of CAV2 becomes $a_{2,N_p}=0$. Based on this, the terminal constraints (\ref{eq:lowerconst8}) and (\ref{eq:lowerconst9}) become equivalent to constraining the terminal state as a steady state: $x_{2,N_p}=f(x_{2,N_p},\gamma_{2,N_p},0)$, which is a concept widely used in economic MPC \cite{amrit2011economic,ferramosca2010economic,rawlings2012fundamentals}. We futher define the optimal reachable steady state \cite{ferramosca2010economic}:

\begin{defn}
The optimal reachable steady state and input, $(x_{i,s},\gamma_{i,s})$, satisfy
\begin{align*}
(x_{i,s}, \gamma_{i,s}) &= \arg \min_{x_i,\gamma_i} l(x_i,\gamma_i) \\
\text{subject to} \quad & x_i = f(x_i, \gamma_i, 0), \\
& x_i \in R_{N_p,i}(x_{i,0}), \gamma_i \in G_{b}.
\end{align*}
\end{defn}
\noindent where $l(\cdot)$ is the stage cost, $f(\cdot)$ is the  discrete platoon system dynamics model, $R_{N_p,i}(x_{i,0})$ and $G_b$ are the $N_p$-step reachable set \cite{borrelli2017predictive} of $x_{i,0}$ and the boundary of input, respectively. We notate the stage cost corresponding to $(x_{i,s},\gamma_{i,s})$ as $\bar{l}(x_{i,0})=l(x_{i,s},\gamma_{i,s})$. The propositions related to the asymptotic local stability of CAV2 can be obtained, details as follow.

\begin{prop} \label{prop asymptotic 1}
Under the assumption that the leading vehicle CAV1 runs at a constant speed and the safety cost is inactive, for CAV2, for any $\varepsilon>0$, there exists a finite lower bound $\hat{\beta}(\epsilon)={\sum_{k=0}^{N_p-1} \left[ \alpha_l \sum_{j=0}^k \alpha_f^{(k-j)} (\|\gamma_{max} - \gamma_{min}\|_2) \right]}/{\epsilon}$ such that, if $\beta\geq\hat{\beta}(\epsilon)$, then for the optimal state trajectory $X_{2}^*=[x_{2,0}^*,\cdots,x_{2,N_p}^*]$ and control sequence $\Gamma_{2}^*=[\gamma_{2,0}^*,\cdots,\gamma_{2,N_p}^*]$, $l(x_{2,N_p}^*,\gamma_{2,N_p}^*)\leq\bar{l}(x_{2,0})+\varepsilon$.
\end{prop}
\begin{proof}
Let $\bar{X}_2=[\bar{x}_{2,0},\cdots,\bar{x}_{2,N_p}]$, $\bar{\Gamma}_2=[\bar{\gamma}_{2,0},\cdots,\bar{\gamma}_{2,N_p}]$ be the state trajectory and control sequence that satisfy the constraints (\ref{eq:lowerconst1})-(\ref{eq:lowerconst9}) and drives CAV2 to the optimal reachable steady state and input $(x_{2,s},\gamma_{2,s})$, the associated cost becomes:
\begin{align*}
J_2(x_{2,0},\bar{\Gamma}_2)=\sum_{k=0}^{N_p-1} l(\bar{x}_{2,k},\bar{\gamma}_{2,k})
+\beta \bar{l}(x_{2,0})
\end{align*}

Considering an arbitrary state trajectory and control sequence $\hat{X}_2=[\hat{x}_{2,0},\cdots,\hat{x}_{2,N_P}]$ and $\hat{\Gamma}_2=[\hat{\gamma}_{2,0},\cdots,\hat{\gamma}_{2,N_p}]$ that satisfy the constraints (\ref{eq:lowerconst1})-(\ref{eq:lowerconst9}) and drive CAV2 into a terminal state satisfying the steady state terminal constraint $f(\hat{x}_{2,N_p}, \hat{\gamma}_{2,N_p}, 0) = \hat{x}_{2,N_p}$, and $l(\hat{x}_{2,N_p}, \hat{\gamma}_{2,N_p}) > \bar{l}(x_{2,0}) + \epsilon$, the associated cost becomes:
\[J_{2}(x_{2,0},\hat{\Gamma}_2) = \sum_{k=0}^{N_p - 1} l(\hat{x}_{2,k}, \hat{\gamma}_{2,k}) + \beta l(\hat{x}_{2,N_p}, \hat{\gamma}_{2,N_p})\]
Taking the cost difference of the two state-control sequences above, we have:
\begin{equation}\label{local stable1}
\begin{split}
J_2(x_{2,0}, \bar{\Gamma}_2) &- J_2(x_{2,0},\hat{\Gamma}_2) = \beta \left[\bar{l}(x_{2,0}) - l(\hat{x}_{2,N_p}, \hat{\gamma}_{2,N_p})\right] \\&+\sum_{k=0}^{N_p - 1} \left[l(\bar{x}_{2,k},\bar{\gamma}_{2,k}) - l(\hat{x}_{2,k}, \hat{\gamma}_{2,k})\right]
\end{split}
\end{equation}

An upper bound of this cost difference can be calculated. We start with the upper bound of the second term, as the states and control of CAV2 are bounded, and the dynamics $f(\cdot)$ and the stage cost $l(\cdot)$ are both continuous. Within the bounded set defined by state and control constraints, both $f(\cdot)$ and $l(\cdot)$ satisfy the local Lipschitz condition \cite{khalil2002nonlinear}:
\begin{equation*}
|l(x_a,\gamma_a )-l(x_b,\gamma_b)| \leq \alpha_l (\|(x_a,\gamma_a )-(x_b,\gamma_b )\|_2)
\end{equation*}
and 
\begin{equation*}
\|f(x_a,\gamma_a )-f(x_b,\gamma_b )\|_2 \leq \alpha_f (\|(x_a,\gamma_a )-(x_b,\gamma_b )\|_2) 
\end{equation*}
where 
$\alpha_l=\underset{(x,\gamma)\in B_r}{\max} \|\nabla l(x,\gamma)\|_2=\underset{(x,\gamma)\in B_r}{\max} \| \begin{bmatrix}2Q_{lon}x \\ 2R_{lon}\gamma\end{bmatrix}\|_2= 2r \rho_r(blockdiag(Q_{lon},R_{lon}))$, $\alpha_f= \| [A B]\|_2$. $B_r$ is a ball with radius $r=\sqrt{\Delta d_r^2+\Delta v_r^2+a_r^2+\gamma_r^2}$, in which the variables $\Delta d_r=max(|\Delta d_{min}|,|\Delta d_{max}|), \Delta v_r=\max(|\Delta v_{min}|,|\Delta v_{max}|)$ where $\Delta v_{min}$ and $\Delta v_{max}$ are the minimum and maximum possible bound of speed difference yielded by (\ref{eq:lowerconst4}), $a_r=\max(|a_{min}|,|a_{max}|)$, and $\gamma_r=\max(|\gamma_{min}|,|\gamma_{max}|)$. $\rho_r(blockdiag(Q_{lon},R_{lon})$ denotes the largest absolute value of eigenvalue (i.e., spectral radius) of $\begin{bmatrix}
    Q_{lon} & 0\\
    0 & R_{lon}
\end{bmatrix}$.

Based on the calculated Lipschitz constants $\alpha_l$ and $\alpha_f$, we can now calculate the upper bound of the second term of the right hand side of (\ref{local stable1}) \cite{fagiano2013generalized}:

\noindent when $k=0$,
\begin{align*}
l(\bar{x}_{2,0}, \bar{\gamma}_{2,0}) - l(\hat{x}_{2,0}, \hat{\gamma}_{2,0}) \leq |l(\bar{x}_{2,0}, \bar{\gamma}_{2,0}) - l(\hat{x}_{2,0}, \hat{\gamma}_{2,0})|\\ 
\leq \alpha_l \|\bar{\gamma}_{2,0} - \hat{\gamma}_{2,0}\|_2 = \alpha_l \sum_{j=0}^{0} \alpha_f^{(k-j)}(\|\bar{\gamma}_{2,j} - \hat{\gamma}_{2,j}\|_2)
\end{align*}
\noindent when $k=1,$
\begin{align*}
&~{}~{}l(\bar{x}_{2,1}, \bar{\gamma}_{2,1}) - l(\hat{x}_{2,1}, \hat{\gamma}_{2,1})\\& \leq |l(\bar{x}_{2,1}, \bar{\gamma}_{2,1}) - l(\hat{x}_{2,1}, \hat{\gamma}_{2,1})| \\
&\leq \alpha_l \|(\bar{x}_{2,1},\bar{\gamma}_{2,1}) - (\hat{x}_{2,1}, \hat{\gamma}_{2,1})\|_2 \\
&\leq \alpha_l \|\bar{x}_{2,1} - \hat{x}_{2,1}\|_2 + \alpha_l \|\bar{\gamma}_{2,1} - \hat{\gamma}_{2,1}\|_2 \\
&\leq \alpha_l (\alpha_f \|\bar{\gamma}_{2,0} - \hat{\gamma}_{2,0}\|_2) + \alpha_l \|\bar{\gamma}_{2,1} - \hat{\gamma}_{2,1}\|_2 \\
&= \alpha_l \sum_{j=0}^{1} \alpha_f^{(k-j)}(\|\bar{\gamma}_{2,j} - \hat{\gamma}_{2,j}\|_2)
\end{align*}

\noindent for generic k, it can be derived with recursion that
\begin{align*}
l(\bar{x}_{2,k}, \bar{\gamma}_{2,k}) - l(\hat{x}_{2,k}, \hat{\gamma}_{2,k}) &\leq \alpha_l \sum_{j=0}^{k} \alpha_f^{(k-j)} \|\bar{\gamma}_{2,j} - \hat{\gamma}_{2,j}\|_2 \\
&\leq \alpha_l \sum_{j=0}^{k} \alpha_f^{(k-j)} \|\gamma_{max} - \gamma_{min}\|_2
\end{align*}
\noindent substituting this to (\ref{local stable1}), we obtain:
\begin{align*}
J_2(x_{2,0}, \bar{\Gamma}_{2}) - J_2(x_{2,0}, \hat{\Gamma}_{2})
\leq \beta \left[ \bar{l}(x_{2,0}) - l(\hat{x}_{2,N_p}, \hat{\gamma}_{2,N_p}) \right] \\
+ \sum_{k=0}^{N_p-1} \left[ \alpha_l \sum_{j=0}^k \alpha_f^{(k-j)} (\|\gamma_{max} - \gamma_{min}\|_2) \right]
\end{align*}
\noindent by how we defined $\bar{x}_2,\bar{{\gamma}_2}$ and $\hat{x}_2,\hat{{\gamma}_2}$, we further have:
\begin{align*}
&~{}~{}J_2(x_{2,0}, \bar{\Gamma}_{2}) - J_2(x_{2,0}, \hat{\Gamma}_{2})\\
&< -\beta \varepsilon + \sum_{k=0}^{N_p-1} \left[ \alpha_l \sum_{j=0}^k \alpha_f^{(k-j)} (\|\gamma_{max} - \gamma_{min}\|_2) \right]
\end{align*}

Based on this upper bound, by setting $\beta\geq\hat{\beta}(\varepsilon)={\sum_{k=0}^{N_p-1} \left[ \alpha_l \sum_{j=0}^k \alpha_f^{(k-j)} (\|\gamma_{max} - \gamma_{min}\|_2) \right]}/{\epsilon}$, it is guaranteed that $J_2(x_{2,0}, \bar{\Gamma}_{2}) - J_2(x_{2,0}, \hat{\Gamma}_{2})<0$. However, it is obvious that the optimization problem defined in subsection \ref{sec4.2} is convex when the safety cost is inactive and a global optimal control sequence $\Gamma_2^*=[\gamma_{2,0}^*,\cdots,\gamma_{2,N_p}^*]$ exists, such that $J_2(x_{2,0}, \bar{\Gamma}_{2}) - J_2(x_{2,0}, \Gamma_2^*)\geq0$. By the contradictory of cost difference, the nature of $\Gamma_2^*$ contradicts the definition and attributes of $\hat{\Gamma}_{2}$. Further by how we defined $\hat{\Gamma}_{2}$, it can be proved that $l(x_{2,N_p}^*,\gamma_{2,N_p}^*)\leq\bar{l}(x_{2,0})+\varepsilon$.
\end{proof}

\begin{prop} \label{prop asymptotic 2}
Under the assumption that the leading vehicle CAV1 runs at a constant speed and the safety cost is inactivate, CAV2 is asymptotically locally stable if $\varepsilon\rightarrow0$, $\beta\geq\hat{\beta}(\varepsilon)$, and $N_p\geq3$.
\end{prop}
\begin{proof}
Under the assumption, the terminal state constraints (\ref{eq:lowerconst8}) and (\ref{eq:lowerconst9}) of CAV2 become $\Delta v_{2,N_p}=0$ and $a_{2,N_p}=0$. Therefore, at time t+1, there exists a feasible but not necessarily optimal control sequence $\Gamma_2(t+1)=[\gamma_{2,1}^*(t), \gamma_{2,2}^*(t), \cdots, \gamma_{2,N_p}^*(t),0]$, where $\gamma_{2,1}^*(t), \gamma_{2,2}^*(t), \cdots, \gamma_{2,N_p}^*(t)$ are the last $N_p$ elements of the optimal control sequence at time t, that drives CAV2 to the terminal state that is identical to the optimal terminal state at time t, $x_{2,N_p}^*(t)$.

In a worst-case scenario, the optimal terminal states remains unchanged until time $t+N_p-1$. At this time, the optimal control sequence $\Gamma_2^*(t+N_p-1)=[\gamma_{2,N_p}^*(t), 0, 0, \cdots, 0]$ ensures the initial state at time $t+N_p$ satisfies $x_{2,0}(t+N_p)=x_{2,N_p}^*(t)$. At time $t+N_p$, if $N_p \geq 3$, let $\gamma_f$ be any feasible control input that has the same sign as $\Delta d_{2,0}(t+N_p)$, there exists a feasible control sequence: $[-0.5\gamma_f, \gamma_f, -0.5\gamma_f, 0,0, \cdots,0]$ that leads to a terminal state $x_{2,N_p}(t+N_p)$ satisfying $l(x_{2,N_p}(t+N_p),\gamma_{2,N_p}(t+N_p))<l(x_{2,N_p}^*(t),\gamma_{2,N_p}^*(t))$ and the terminal constraints (\ref{eq:lowerconst8})-(\ref{eq:lowerconst9}).

Therefore, for every $N_p$ steps, when $N<N_p$, the terminal stage cost satisfies $l(x_{2,N_p}^*(t+N),\gamma_{2,N_p}^*(t+N))\leq l(x_{2,N_p}^*(t),\gamma_{2,N_p}^*(t))$. This indicates CAV2 is Lyapunov local stable by choosing the terminal stage cost as the Lyapunov function. When $N=N_p$, by applying Proposition \ref{prop asymptotic 1}, if $\varepsilon\rightarrow0$, the terminal stage cost satisfies $l(x_{2,N_p}^*(t+N_p),\gamma_{2,N_p}^*(t+N_p))< l(x_{2,N_p}^*(t),\gamma_{2,N_p}^*(t))$, and the terminal spacing deviation is guaranteed to decrease. Following this decrement of the terminal spacing deviation until the equilibrium state lies within the reachable set of the MPC controller, the controller will behave as an MPC controller with terminal state constrained to 0 \cite{fagiano2013generalized} and eventually converge to the equilibrium state, indicating that CAV2 is asymptotically locally stable.
\end{proof}

With the asymptotic local stability of CAV2 proven, we can further prove the asymptotic local stability conditions of the rest of the followers.

\begin{prop}
Under the assumption that the leading vehicle CAV1 runs at a constant speed and the safety cost is inactive, for $i\geq3$, CAV$i$ is asymptotically locally stable if $\varepsilon\rightarrow0$, $\beta\geq\hat{\beta}(\varepsilon)$, and $N_p\geq3$.
\end{prop}

\begin{proof}
By Proposition \ref{prop asymptotic 2}, the first follower, CAV2, is asymptotically locally stable. Therefore, CAV2 will converge to the equilibrium state and run at a constant speed. After CAV2 has converged, CAV3 can be treated the same as CAV2, and the proof of asymptotic local stability follows a process similar to the proof of Proposition \ref{prop asymptotic 1} and Proposition \ref{prop asymptotic 2}. Again, after CAV3 converges, the asymptotic local stability of the rest of the followers can be proved in the same manner.
\end{proof}

\begin{rem}
    The Propositions in this subsection require the leading vehicle to maintain a constant velocity, and the safety cost to be inactive. The constant-speed requirement aligns with the conditions in studies \cite{dunbar2011distributed,li2020distributed,zhou2019distributed}. Practically, this assumption is justifiable as the main objective of the leading vehicle is to operate relatively constantly at the free-flow speed. The safety cost-inactive assumption is also justifiable as the cost only becomes active when the following three conditions are all satisfied: (1) the predecessor-follower pair drive on the same road at at least one time step within the prediction horizon; (2) The spacing deviation satisfies $\Delta d_{i,0}\leq -\Delta d_{safe}$; (3) The speed difference satisfies $\Delta v_{i,0}\leq 0$. Therefore, the safety cost only becomes active in extreme situations (e.g., extreme initial states, emergency breaking of the predecessor) and does not effect the ordinary following situation. Furthermore, when the safety cost becomes active, the CAV operates in a safety mode where local stability becomes far less important than ensuring safety.
\end{rem}

\begin{rem} \label{remark asymptotic}
The propositions in this subsection also implicitly assume initial feasibility and recursive feasibility. Recall that constraint (\ref{eq:lowerconst3}) only sets theoretical boundaries on the spacing deviation for the proof in Proposition \ref{prop asymptotic 1}, its impact on feasibility can be ignored by setting a small enough $\Delta d_{min}$ and a large enough $d_{max}$. Constraint (\ref{eq:lowerconst4}) is essentially equivalent to constraining the velocity within $[v_{min},v_{max}]$. Given initial feasibility, when the leading vehicle runs at a constant speed, it is not difficult to verify that there exists $\Gamma_i(t+1)=[\gamma_{i,1}^*(t),\gamma_{i,2}^*(t),\cdots,\gamma_{i,N_p}^*(t),0]$, which is a feasible solution at time $t+1$, indicating recursive feasibility.
\end{rem}

 A more detailed analysis on initial feasibility will be presented in subsection \ref{sec4.5}.

\subsection{String Stability Analysis} \label{sec4.4}
In this subsection, we provide the string stability (disturbance attenuation through a vehicular string) analysis of the controller proposed in subsection \ref{sec4.2}. Specifically, mathematical definition, conditions, and proof of string stability is provided.

We consider the $l_2$ norm string stability defined as follow \cite{zhou2019distributed,naus2010string}.

\begin{defn}
For a system (a CAV platoon with length N), it is $l_2$ norm string stable if and only if

$\frac{\|\Delta d_{i+1}(t)\|_{l2}}{\|\Delta d_{i}(t)\|_{l2}} \leq 1$  for $\forall i \in \{1, 2, \ldots, N-1\}$
\end{defn}

We first derive the explicit solution of the optimization problem proposed in subsection \ref{sec4.2} when the constraints (\ref{eq:lowerconst3})-(\ref{eq:lowerconst9}) and the safety cost are inactive (i.e., the constraints don't effect the control law and the CAV does not operate in the safety mode), which we refer as the finite time unconstrained optimal control (FTUOC) problem. Note that the FTUOC differs from the standard linear quadratic regulator (LQR) due to the external input of the predecessor's optimal acceleration sequence in (\ref{eq:lowerconst2}). So the explicit feedback gain of the standard finite-time LQR, which are derived by solving the discrete algebraic Riccati equation, is not the exact solution of the FTUOC. However, we can slightly modify the batch approach in \cite{borrelli2017predictive} and yield a formula for the optimal control sequence as a function of the initial state $x_{i,0}$ and the predecessor's optimal acceleration sequence $\Tilde{A}_{i-1}$, details as follow.

Given the initial state $x_{i,0}$ and the predecessor's optimal acceleration sequence, all future states can be explicitly expressed as a function of the future inputs:
\begin{equation} \label{OC model}
\Tilde{X}_i=S_xx_{i,0}+S_u\Tilde{\Gamma}_i+S_a\Tilde{A}_{i-1}
\end{equation}

\noindent where $\Tilde{X}_i=
\begin{bmatrix}
x_{i,0} \\
x_{i,1}\\
x_{i,2}\\
\vdots\\
x_{i,N_p}
\end{bmatrix},
S_x=\begin{bmatrix}
    I\\
    A\\
    A^2\\
    \vdots\\
    A^{N_p}
\end{bmatrix},
\Tilde{\Gamma_i}=\begin{bmatrix}
    \gamma_{i,0}\\
    \gamma_{i,1}\\
    \gamma_{i,2}\\
    \vdots\\
    \gamma_{i,N_p}
\end{bmatrix},$
$
S_u=\begin{bmatrix}
    0 & 0 & \cdots & 0 & 0\\
    B & 0 & \cdots & 0 & 0\\
    AB & B & \cdots & 0 & 0\\
    \vdots & \vdots &\ddots & \vdots & 0\\
    A^{N_p-1}B & A^{N_p-2}B & \cdots & B & 0
\end{bmatrix},
S_a=\begin{bmatrix}
    0 & 0 & \cdots & 0 & 0\\
    D & 0 & \cdots & 0 & 0\\
    AD & D & \cdots & 0 & 0\\
    \vdots & \vdots &\ddots & \vdots & \vdots\\
    A^{N_p-1}D & A^{N_p-2}D & \cdots & D & 0
\end{bmatrix},
\Tilde{A}_{i-1} = \begin{bmatrix}
    a_{i-1,0}^*\\
    a_{i-1,1}^*\\
    a_{i-1,2}^*\\
    \vdots\\
    a^*_{i-1,N_{p}}\\
\end{bmatrix}
$.

With $\Tilde{Q}=blockdiag\{Q_{lon},Q_{lon},\cdots,\beta Q_{lon}\}, \Tilde{R}=[R_{lon},R_{lon},\cdots,\beta R_{lon}]$, substituting (\ref{OC model}) into the cost function (\ref{eq:lowercost}):

\begin{equation} \label{OC cost}
J_i(x_{i,0},\Tilde{\Gamma_i})=\Tilde{\Gamma_i}^TH\Tilde{\Gamma_i}+2x_{i,0}^TF\Tilde{\Gamma_i}+2\Tilde{A}_{i-1}^TL\Tilde{\Gamma_i}+f_J(x_{i,0},\Tilde{A}_{i-1})
\end{equation}

\noindent where $H=(S_u^T\Tilde{Q}S_u+\Tilde{R}), F=S_x^T\Tilde{Q}S_u$, and $L=S_a^T\Tilde{Q}S_u$, $f_J(x_{i,0},\Tilde{A}_{i-1})$ is omitted as it does not effect the rest of the deduction.

Taking the gradient of (\ref{OC cost}) and setting it to zero, we obtain the optimal control sequence as:
\begin{equation}\label{explicit solution}
\Gamma_i^*=-H^{-1}F^Tx_{i,0}-H^{-1}L^T\Tilde{A}_{i-1}
\end{equation}

As only the first control input is executed at each time step, we can obtain the linear feedback and feedforward gain from the formula of the first element of $\Gamma_i^*$:
\begin{equation}\label{linear gain}
    \gamma_{i,0}^*=K_bx_{i,0}+K_f\Tilde{A}_{i-1}
\end{equation}

\noindent where $K_b=[k_{\Delta d}~{} k_{\Delta v}~{} k_a$] is the linear feedback gain, $k_{\Delta d}, k_{\Delta v}$ and $k_a$ are the gains corresponding to spacing deviation, speed difference, and acceleration, respectively. For the linear feedforward gain, $K_f=[k_{f,0},k_{f,1},\cdots,k_{f,N_p}]$, where $k_{f,j}$ is the gain corresponding to $a_{i-1,j}$ for $j=0,1,\cdots,N_p$. 

Based on the explicit discrete linear feedback and feedforward gain in (\ref{linear gain}), given that the discretization time step is sufficiently small, we have $a_{i-1,j}\approx a_{i-1,0}$ for $j=1,2,\cdots,N_p$ and the discrete gains $K_b$ and $K_f$ can be approximately treated as the continuous gains \cite{kleinman1977continuous,melzer1971sampling}. Therefore, for the continuous system (\ref{eq:cd1})-(\ref{eq:cd3}), the explicit solution of the optimization problem (\ref{eq:lowercost})-(\ref{eq:lowerconst9}) is given as:
\begin{equation} \label{continuous linear gain}
\begin{split}
    \gamma_i(t)&=g_i(T_{des,i}(t), T_i(t), v_i(t))\\
    &=K_b x_{i}(t) + k_fa_{i-1}(t)
\end{split}
\end{equation}

\noindent where $\gamma_i(t)$ is the continuous control input of CAV$i$, and $k_f=\sum_{j=0}^{N_p} k_{f,j}$.

With the explicit formula of the continuous control input $\gamma_i(t)$, we can derive the conditions on the feedback and feedforward gains that make the CAV platoon $l_2$ norm string stable. By the Parseval theorem \cite{10.1112/plms/s2-45.1.458}, energy is the same in the frequency and time domain, and $\|G_i(s)\|_{h_\infty}=\sup_\omega(\frac{\|\Delta d_i(s)\|_2}{\|\Delta d_{i-1}(s)\|_2})\leq1$ is a sufficient condition for $l_2$ norm string stable \cite{zhou2019distributed}, where $\Delta d_i(s)$ is the Laplace transform of $\Delta d_i(t)$ and $G_i(s)$ is the transfer function, $\|G_i(s)\|_{h_\infty}$ is the $h$ infinity norm of the transfer function.

\begin{prop}
    The proposed lower-level longitudinal MPC controller is $l_2$ norm string stable when the constraints (\ref{eq:lowerconst3})-(\ref{eq:lowerconst9}) and the safety cost are inactive and the following condition holds:
    \begin{equation*}
    p^2-q\leq0
    \end{equation*}
    
    \noindent or: 
    \begin{equation*}
    \begin{tabular}{l}
    $\left\{\begin{array}{l}
    p^2-q>0 \\
    \frac{-p\pm\sqrt{p^2-q}}{2}\leq0 \\
    \end{array} \right.$
    \end{tabular}
    \end{equation*}
    
    \noindent where $p=k_a^2-k_f^2-2k_{\Delta v}$, $q=8k_{\Delta d}(k_a+k_f)$.
\end{prop}
\begin{proof}
    By (\ref{continuous linear gain}), the control input difference between CAV$i$ and its predecessor follows the relation:
    \begin{equation} \label{laplace 1}
    \begin{split}
        \gamma_{i-1}(t) - \gamma_i(t) =& K_b(x_{i-1}(t)-x_i(t))\\
        &+ k_f(a_{i-2}(t) - a_{i-1}(t))
    \end{split}
    \end{equation}
    
    Taking the Laplace transform of (\ref{laplace 1}) and rearranging it:
    \begin{equation} \label{laplace 2}
    \begin{split}
        &s^3 \Delta d_i(s) - k_a s^2 \Delta d_i(s) + k_{\Delta d} \Delta d_i(s) + k_{\Delta v} s \Delta d_i(s) \\
        &= k_f s^2 \Delta d_{i-1}(s) + k_{\Delta v} s \Delta d_{i-1}(s) + k_{\Delta d} \Delta d_{i-1}(s)  
    \end{split}
    \end{equation}

    Based on (\ref{laplace 2}), the transfer function is obtained by:
    \begin{equation*}
        G_i(s) = \frac{\Delta d_i(s)}{\Delta d_{i-1}(s)} = \frac{k_f s^2 + k_{\Delta v} s + k_{\Delta d}}{s^3 - k_a s^2 + k_{\Delta v} s + k_{\Delta d}}
    \end{equation*}

    The magnitude of the transfer function can be calculated by substituting $s=j\omega$:
    \begin{equation*}
        |G_i(j\omega)| = \sqrt{\frac{(-\omega^2 k_f + k_{\Delta d})^2 + \omega^2 k_{\Delta v}^2}{(\omega^2 k_a + k_{\Delta d})^2 + (-\omega^3 + k_{\Delta v}\omega)^2}}
    \end{equation*}
    Regulating $\|G_i(j\omega)\|_{h_\infty}\leq1$ yields:
    \begin{equation} \label{laplace 5}
        \omega^6 + (k_a^2 - k_f^2 - 2k_{\Delta v}) \omega^4 + 2k_{\Delta d}(k_a + k_f) \omega^2 \geq 0
    \end{equation}

    Note that $\omega$ here is the frequency of oscillation, and $\omega\neq0$. Dividing both side of (\ref{laplace 5}) with $\omega^2$ and conducting a change of variables yields an equivalent inequality expression:
    \begin{equation}\label{laplace 6}
        W^2 + (k_a^2 - k_f^2 - 2k_{\Delta v})W + 2k_{\Delta d}(k_a + k_f) \geq 0
    \end{equation}
    where $W=\omega^2>0$, the left-hand-side is a quadratic function of $W$ with a positive leading coefficient, to satisfy (\ref{laplace 6}), the quadratic function should either have no root, 1 root, or two non-positive roots.

    Defining $p=k_a^2-k_f^2-2k_{\Delta v}$ and $q=8k_{\Delta d}(k_a+k_f)$, based on the roots of quadratic equation, when there exists no root or only one root:
    \begin{equation} \label{string stable condition1}
    p^2-q\leq0
    \end{equation}
    
    \noindent when there are two non-positive roots: 
    \begin{equation} \label{string stable consition2}
    \begin{tabular}{l}
    $\left\{\begin{array}{l}
    p^2-q>0 \\
    \frac{-p\pm\sqrt{p^2-q}}{2}\leq0 \\
    \end{array} \right.$
    \end{tabular}
    \end{equation}
\end{proof}

The conditions (\ref{string stable condition1}) and (\ref{string stable consition2}) for the CAV platoon to be $l_2$ norm stable are visualized in the $p$-$q$ space, shown in Fig. \ref{fig:string region}. The boundary of the two regions consist of two parts: $q=p^2$ for $p\leq0$, and $q=0$ for $p>0$.
\begin{figure}[h]
\centering
   {\includegraphics[width=0.25\textwidth]{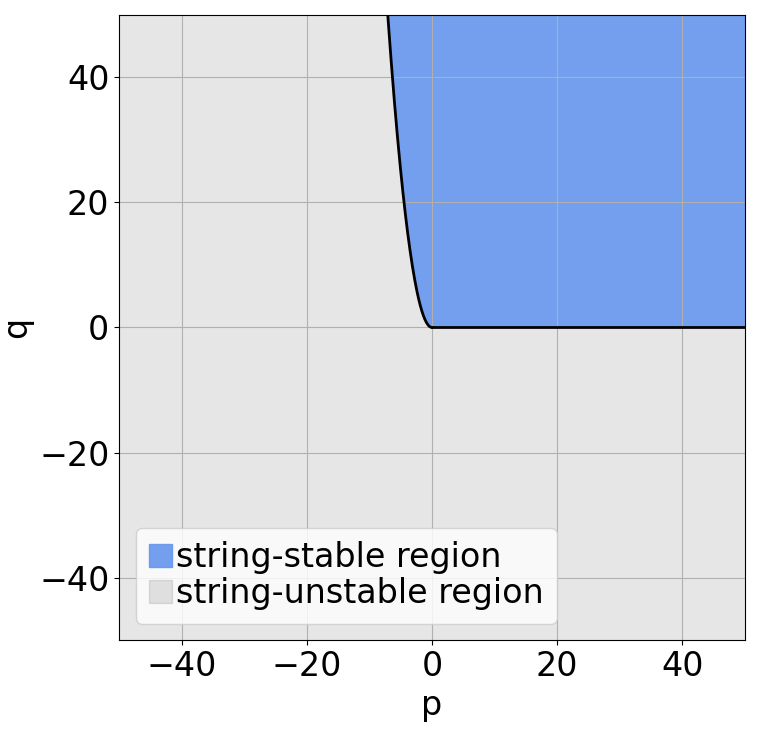}
    \caption{String-stable and string-unstable regions in $p$-$q$ space}
    \label{fig:string region}}
\end{figure}

\subsection{Initial Feasibility Analysis} \label{sec4.5}
As stated in Remark \ref{remark asymptotic}, the proof of asymptotic local stability implicitly assumes initial feasibility. In this subsection, initial feasibility analysis of the controller proposed in subsection \ref{sec4.2} is presented. Specifically, definition of initial feasible set of MPC is introduced, accompanied by a comparison of the proposed controller's initial feasible set with other commonly used local-stable MPC controllers.

We first introduce the definition of precursor set, which is an important concept for the initial feasible set of MPC \cite{borrelli2017predictive}:
\begin{defn}
    For an autonomous system $x(t+1)=f(x(t),\gamma(t))$ subjected to constraints $x(t)\in \chi_{b}$, $\gamma(t)\in G_b$, $\forall t\geq 0$, the precursor set to the set $S$ is denoted as
    $Pre(S)=\{x\in \mathbb{R}^n: \exists\gamma\in G_b ~{}s.t.~{} g(x,u)\in S\}$.
\end{defn}
\noindent $Pre(S)$ is the set of states which can be driven into the
target set $S$ in one time step while satisfying input constraints.

Another definition that is closely related to the initial feasibility of MPC is the N-step controllable set\cite{borrelli2017predictive}:

\begin{defn}
    For an autonomous system $x(t+1)=g(x(t),\gamma(t))$ subjected to constraints $x(t)\in \chi_b$, $\gamma(t)\in G_b$, $\forall t\geq 0$, for a given set $S\in \chi_b$, the N-step controllable set $\kappa_N(S)$ is defined recursively as:
    \begin{equation} \label{N-step controllable}
    \kappa_j(S)=Pre(\kappa_{j-1}(S))\cap \chi_b, \kappa_0(S)=S
    \end{equation}
\end{defn}
\noindent All states $x_0$ of the autonomous system belonging to the N-Step Controllable
Set $\kappa_N(S)$ can be driven, by a suitable control sequence, to the target set S in N
steps, while satisfying input and state constraints\cite{borrelli2017predictive}.

Now, we can define the initial feasible set \cite{borrelli2017predictive}:
\begin{defn}
    For the following problem:
    \begin{equation*}
    \begin{split}
    J^*(x(0))&=min_{\Gamma} J(x(0),\Gamma)  \\
    \text{subject to}~{}~{} &x_{k+1}=f(x_k,\gamma_k), ~{}k=0,\cdots,N_p-1\\
    &x_k\in \chi_b,~{}\gamma\in G_b,~{}k=0,\cdots,N_p-1\\
    &x_{N_p}\in \chi_f\\
    &x_0=x(0)
    \end{split}
    \end{equation*}
    \noindent where $N_p$ is the prediction horizon, $\chi_f$ is the terminal constraint set. The initial feasible set is defined as:
    
    $\chi_0=\{x_0\in \chi_b:~{}\exists(\gamma_0,\cdots,\gamma_{N_p-1})$ such that $x_k\in \chi_b,~{}\gamma_k\in G_b,~{}k=0,\cdots,N_p-1,~{}x_{N_p}\in \chi_f$ where $x_{k+1}=f(x_k,u_k),~{}k=0,\cdots,N_p-1\}$.
\end{defn}

A control sequence $(\gamma_0,\cdots,\gamma_{N_p-1})$ can only be found, (i.e., the control problem is feasible), if $x(0)\in \chi_0$. The initial feasible set can be calculated using the N-step controllable set \cite{borrelli2017predictive}: $\chi_0=\kappa_{N_p}(\chi_f)$.

Based on the definition of initial feasible set, we compare the initial feasible set of the proposed MPC controller with those of other commonly used local-stable MPC controllers (i.e., MPC controllers with zero-terminal constraint and with invariant set-terminal constraints). Without loss of generality, for simplicity and to effectively illustrate the comparison results we investigate the a scenario where the leading vehicle CAV1 runs on the mainline with a constant speed, and its follower CAV2 runs on the on-ramp. The shape of the initial feasible set $\chi_0$ depends on the state boundary $\chi_b$, the control boundary $\Gamma_b$, the controller horizon $N_p$, and the terminal set $\chi_f$\cite{borrelli2017predictive}. We keep $\chi_b$, $G_b$, and $N_p$ the same for all three controllers to show how the proposed controller enhances feasibility with the designed terminal constraints (\ref{eq:lowerconst8}) and (\ref{eq:lowerconst9}). Consequently, this demonstration exemplifies the proposed controller's practicality for the on-ramp merging scenario, particularly in situations where the initial spacing deviation of CAVs tends to be large. Specifically, the boundaries and prediction horizon of CAV2 are set as: $\Delta d_{2,k}\in[-30,30]$, $\Delta v_{2,k}\in[-3,3]$, $a_{2,k}\in[-3,3]$, $\gamma_{2,k}\in[-5,5]$, $k=0,1,\cdots,N_p$, $N_p=10$.

The precursor set here is calculated following\cite{borrelli2017predictive}:
\begin{equation*}
    Pre(S)=(S\oplus(-B)\circ G_b)\circ A
\end{equation*}
where the symbol $\oplus$ represents the Minkowski sum and the symbol $\circ$ represents the linear transformation of sets. However, the detailed deduction of this calculation extends beyond the scope of this paper. For those interested in the calculation details of the precursor set and Minkowski sum, we refer readers to \cite{borrelli2017predictive,halperin2002robust}. The comparison of the initial feasible set at $t=0$ of CAV2 of the zero-terminal constraint MPC \cite{dunbar2011distributed}, the invariant set-terminal constraint MPC \cite{zhou2019distributed}, and the proposed MPC is shown in Fig. \ref{fig:feasible set}.

\begin{figure}[h]
    \centering
    \setlength{\abovecaptionskip}{0pt}
    \subfigure[]{\includegraphics[width=0.2\textwidth]{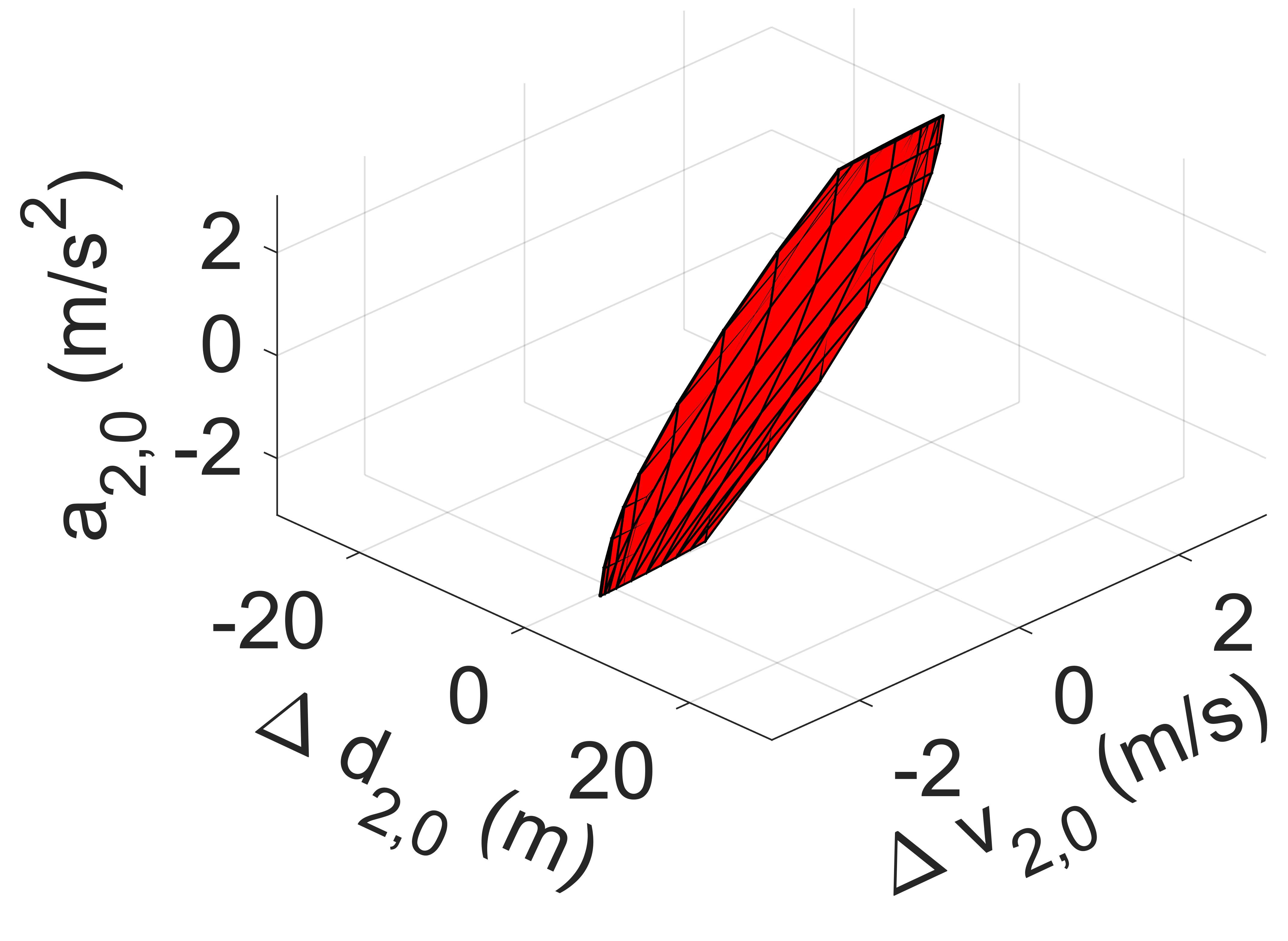}}
    \subfigure[]{\includegraphics[width=0.2\textwidth]{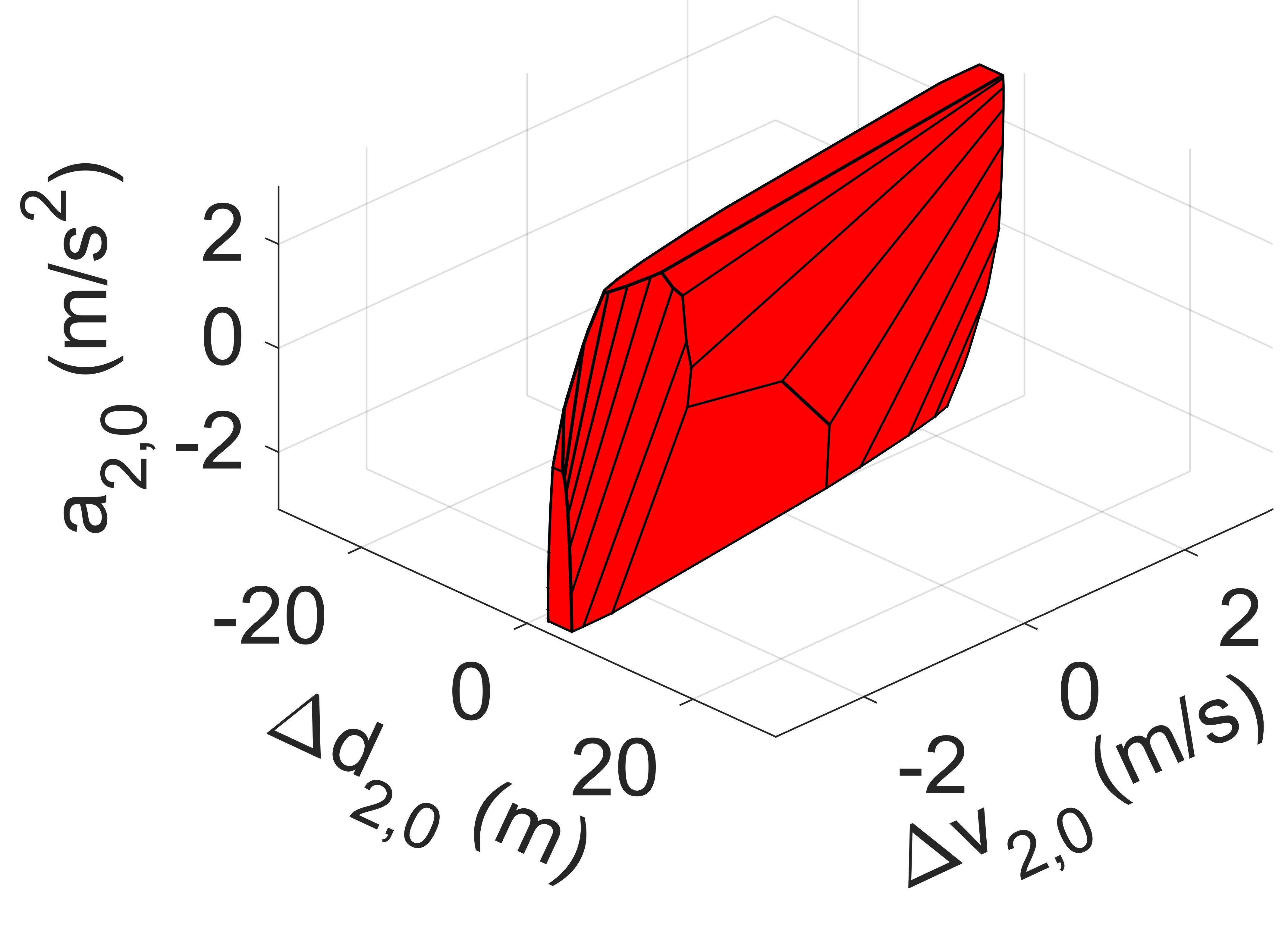}}
    \subfigure[]{\includegraphics[width=0.2\textwidth]{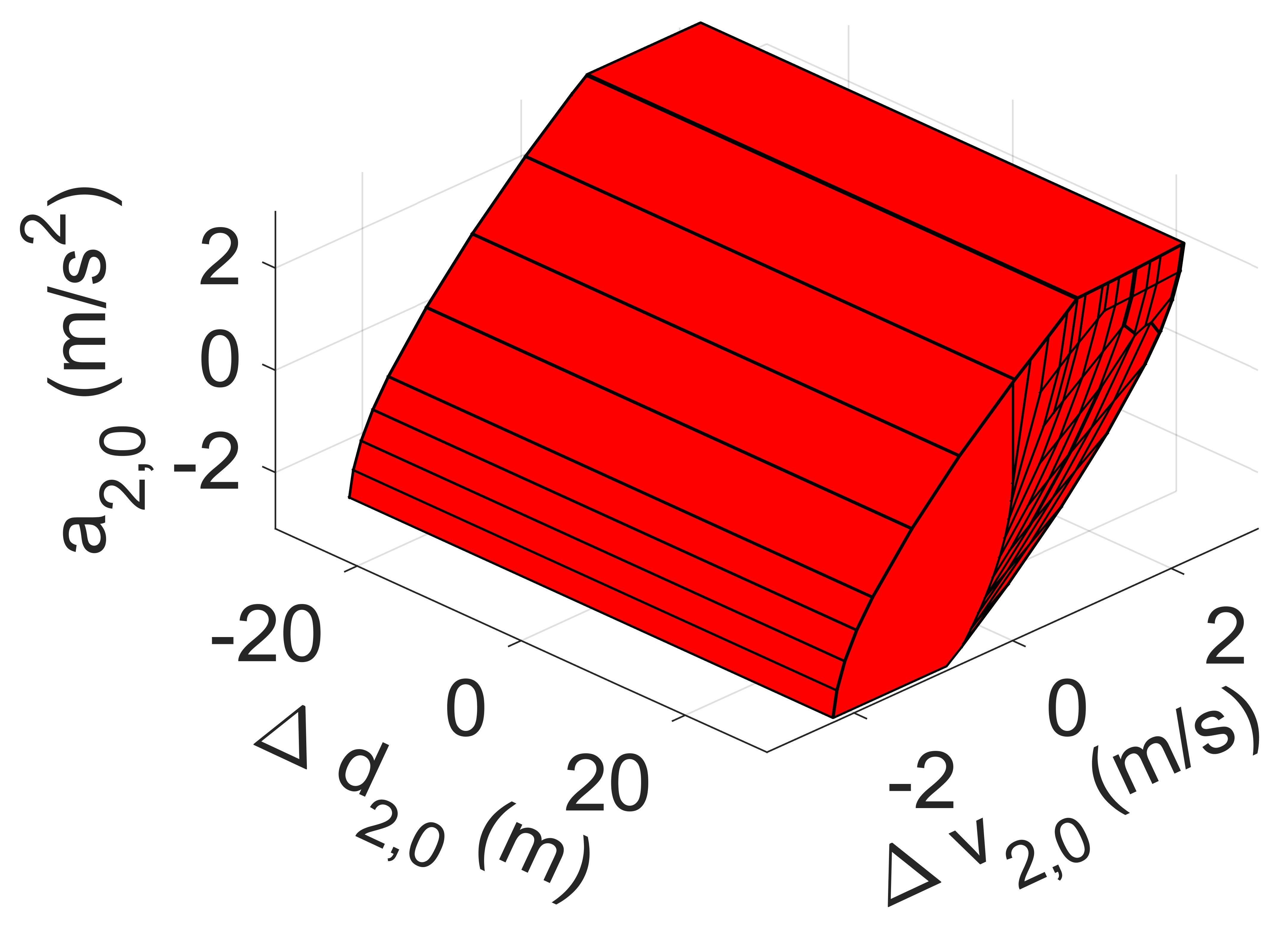}}
    \caption{Comparison of initial feasible sets of local-stable MPC controllers with different terminal constraints. (a) Zero-terminal constrained MPC. (b) Invariant set-terminal constrained MPC. (c) Proposed MPC.}
    \label{fig:feasible set}
\end{figure}

From Fig. \ref{fig:feasible set}, it is observed that the MPC with zero-terminal constraint possesses the smallest initial feasible set, this is due to the fact that its recursive calculation (\ref{N-step controllable}) starts at the origin point. While the initial feasible set of the invariant set-terminal constraint MPC is relatively large, it nevertheless imposes restrictions on all three states. Compared to Fig. \ref{fig:feasible set}(b), the initial feasible set of the proposed MPC, despite having slightly more restrictions on $\Delta v_{2,0}$ and $a_{2,0}$, is considerably expanded on the $\Delta d_{2,0}$ axis. This expansion renders it far more practical for on-ramp merging scenarios, where initial spacing deviation can be exceptionally large.  As highlighted in subsection \ref{sec4.2}, constraint (\ref{eq:lowerconst3}) is employed purely for theoretical boundedness of $\Delta d_{i,k}$. Therefore, by reducing $\Delta d_{min}$ and increasing $\Delta d_{max}$, the initial feasible set can be expanded further along the $\Delta d_{2,0}$ axis without affecting the controller's performance.

\section{LOWER-LEVEL LATERAL CONTROLLER\label{sec5}}
To tackle real-world on-ramp merging, it is crucial to consider curvy ramps and ensure CAVs stay on the centerline with minimal deviation. Therefore, we propose a lower-level lateral controller based on a linear MPC formulation that tracks the centerline and regulates lateral deviation and heading angle deviation. In subsection \ref{sec5.1}, we introduce the kinematic bicycle model. Subsection \ref{sec5.2} presents the optimization formulation for the MPC lateral controller.

\subsection{Kinematic Bicycle Model\label{sec5.1}}
The bicycle model characterizes the continuous nonlinear kinematic behavior of each CAV as follows:
\begin{equation} \label{eq:bycicle}
\dot{\boldsymbol{\chi}} = \begin{bmatrix} 
\dot{X} \\ 
\dot{Y} \\ 
\dot{\theta} 
\end{bmatrix} = \begin{bmatrix} 
v \cos \theta \\
v \sin \theta \\
v \tan \delta/L \end{bmatrix}=
f_n(\chi,\mu)
\end{equation}

\noindent where $\chi= \begin{bmatrix} 
X \ 
Y \
\theta
\end{bmatrix}^T$ represent the state of the CAV, $X$ and $Y$ denote the coordinates of the rear axle's center, $\theta$ signifies the heading angle relative to the $X$-axis. Note that $\chi$ and $X$ here are new variables and differ from the variables from Section \ref{sec4} The control input is given by $\mu=\begin{bmatrix} 
v \ 
\delta 
\end{bmatrix}^T$, where $v$ is the longitudinal velocity and $\delta$ is the steering angle of the CAV. $L$ denotes the CAV's wheelbase.

Using Taylor expansion and neglecting the higher-order terms, we linearize the nonlinear bicycle model around a feasible trajectory:
\begin{equation}\label{eq:linearization}
\begin{split}
\dot{{\chi}} &\approx 
{f}_l({\chi},{\mu},{\chi}_s,{\mu}_s)
\end{split}
\end{equation}

\noindent where $\mu_s=\begin{bmatrix}v_s \ \delta_s\end{bmatrix}^T$ represents a feasible control sequence, and $\chi_s=\begin{bmatrix}X_s \ Y_s \ \theta_s\end{bmatrix}^T$ is the corresponding sequence of states. Further, by employing Euler discretization, we obtain the linear discrete bicycle model:
\begin{equation}\label{eq:dbicycle}
\begin{split}
\chi(k+1)=\chi(k)+{f}_l({\chi(k)},{\mu(k)},{\chi}_s(k),{\mu}_s(k))*T_s
\\=A_d(k)\chi(k)+B_d(k)\mu(k)+D_d(k)
\end{split}
\end{equation}

\noindent where $T_s$ is the time interval between each step, which is consistent with the interval of the lower-level longitudinal controller in section \ref{sec4}. We have matrices $A_d(k)=\begin{bmatrix} 1 & 0 & -v_s(k)T_s\sin(\theta_s(k)) \\ 0 & 1 & v_s(k)T_s\cos(\theta_s(k)) \\ 0 & 0 & 1 \end{bmatrix}$, $B_d(k)=\begin{bmatrix}
T_s\cos{\theta_s (k)} &0 \\
T_s\sin{\theta_s (k)} &0 \\
\frac{T_s\tan{\delta_s (k)}}{L} & \frac{T_sv_s (k)}{L\cos^2{\delta_s (k)}}
\end{bmatrix}$, $D_d(k)=\begin{bmatrix}
T_sv_s(k)\sin(\theta_s(k))X_s(k) \\
-T_sv_s(k)\cos(\theta_s(k))Y_s(k) \\
-\frac{T_sv_s(k)\delta_s(k)}{L\cos^2(\delta_s(k))} \
\end{bmatrix}$.

\subsection{Optimization Formulation\label{sec5.2}}
After the longitudinal controller generates a solution, the lower-level lateral controller is executed on each CAV locally to regulate its lateral deviation and heading angle deviation. The lateral controller utilizes the predicted velocity from the longitudinal controller to more accurately predict the future states of the CAV. Note that due to this dependency, the lateral controller's prediction horizon should be equal or shorter than the longitudinal controller's. For the sake of simplicity, we keep the two prediction horizon the same here.

The optimization problem solved by the MPC controller at each time step is formulated as follows:
\begin{equation}\label{eq:latcost}
\begin{split}
J_{lat}^*(\chi_t, M) = \min_{\mu(0) \rightarrow \mu(N_p)} \sum_{k=0}^{N_p} \mu^T(k) R_{lat} \mu(k) \\ + (\chi(k)-\chi_{ref}(k))^T Q_{lat} (\chi(k)-\chi_{ref}(k))
\end{split}
\end{equation}
\noindent subject to:
\begin{equation}\label{eq:latconst1}
\chi(0)=\chi_t
\end{equation}
\begin{equation}\label{eq:latconst2}
\begin{split}
\chi(k+1)=A_d(k)\chi(k)+B_d(k)\mu(k)+D_d(k)\\
,~{}k\in \{0,1,\cdots,N_p-1\}
\end{split}
\end{equation}
\begin{equation}\label{eq:latconst3}
\delta_{min} \leq \delta(k) \leq \delta_{max}
,~{}k\in \{0,1,\cdots,N_p\}
\end{equation}
\begin{equation}\label{eq:latconst4}
\Delta\delta_{min} \leq \Delta\delta(k) \leq \Delta\delta_{max}
,~{}k\in \{0,1,\cdots,N_p\}
\end{equation}
\begin{equation}\label{eq:latconst5}
v(k)=v_{lon}(k)
,~{}k\in \{0,1,\cdots,N_p\}
\end{equation}

\noindent where (\ref{eq:latcost}) represents the cost function, which penalizes the control inputs and the deviation of the states from the reference waypoints, with the weight matrices $R_{lat}$ and $Q_{lat}$, respectively. (\ref{eq:latconst1}) is the initial state constraint and (\ref{eq:latconst2}) is the discrete kinematic model constraint. (\ref{eq:latconst3}) and (\ref{eq:latconst4}) set boundaries for the steering angle and its increment, respectively. (\ref{eq:latconst5}) restricts the velocity input to be equal to the predicted velocity of the longitudinal controller to better predict the future CAV states.

\section{SIMULATION AND RESULTS\label{sec6}}
To demonstrate the proposed framework's performance, we've assessed its functionality through three numerical simulation scenarios. The local asymptotic stability characteristic and upper-level controller performance, the string stability characteristic, and the lower-level lateral controller performance are exemplified in the three scenarios, respectively. In all three scenarios, the road geometry is constructed with a $400m$ straight mainline, and an on-ramp including a $397.5m$ straight segment followed by a circular arc of $47.75m$ radius for merging into the mainline. Scenarios 1 and 2 primarily focus on the longitudinal controllers and scenario 3 examines the lateral controller.

The simulations are executed on a laptop equipped with an Intel Core i7-10870H processor, the simulation setup parameters are shown in Table \ref{tab:1}.

\begin{table}[h]
\centering
\scriptsize
\setlength{\abovecaptionskip}{0pt}
\caption{Simulation Parameters}
\label{tab:1}
\begin{minipage}[t]{0.45\linewidth}
    \centering
    \begin{adjustbox}{width=\linewidth}
    \begin{tabular}[t]{|>{\centering\arraybackslash}m{1.2cm}|>{\centering\arraybackslash}m{1.2cm}|}
        \hline
        \textbf{\textit{Parameter}} & \textbf{\textit{Value}} \\
        
        \hline
        $\eta_i$ & $0.8$ \\
        \hline
        $m_i$ & $1500kg$ \\
        \hline
        $r_i$ & $0.25m$ \\
        \hline
        $\tau_i$ & $0.4s$ \\
        \hline
        $f_{roll,i}$ & $0.015$ \\
        \hline
        $g$ & $9.8N/kg$ \\
        \hline
        $\rho$ & $1.2kg/m^3$ \\
        \hline
        $C_d$ & $0.25$ \\
        \hline
        $A_{v,i}$ & $2m^2$ \\
        \hline
        $T_s$ & $0.1s$ \\
        \hline
        $\Delta d_{min}$ & $-30m$\\
        \hline
        $\Delta d_{max}$ & $30m$\\
        \hline
        $\Delta d_{safe}$ & $5m$\\
        \hline
    \end{tabular}
    \end{adjustbox}
\end{minipage}
\hspace{0.01\linewidth}
\begin{minipage}[t]{0.45\linewidth}
    \centering
    \begin{adjustbox}{width=\linewidth}
    \begin{tabular}[t]{|>{\centering\arraybackslash}m{1.2cm}|>{\centering\arraybackslash}m{1.2cm}|}
        \hline
        \textbf{\textit{Parameter}} & \textbf{\textit{Value}} \\
        \hline
        $v_{min}$ & $0m/s$ \\
        \hline
        $v_{max}$ & $30m/s$\\
        \hline
        $a_{min}$ & $-5m/s^2$ \\
        \hline
        $a_{max}$ & $5m/s^2$ \\
        \hline
        $\gamma_{min}$ & $-5m/s^3$ \\
        \hline
        $\gamma_{max}$ & $5m/s^3$ \\
        \hline
        $L$ & $2.7m$ \\
        \hline
        $\delta_{min}$ & $-0.8rad$ \\
        \hline
        $\delta_{max}$ & $0.8rad$ \\
        \hline
        $\Delta\delta_{min}$ & $-0.04rad$ \\
        \hline
        $\Delta\delta_{max}$ & $0.04rad$ \\
        \hline
        $N_p$ & $12$ \\
        \hline
    \end{tabular}
    \end{adjustbox}
\end{minipage}
\end{table}

\subsection{Local Stability and Upper-level Controller Experiments\label{sec6.1}}

In this subsection, we examine both the local asymptotic stability of the lower-level longitudinal controller and the effectiveness of the proposed mixed-integer programming (MIP) upper-level controller. It is crucial to present the convergence of states and the efficiency of the trajectory evolution, hence a scenario is designed with large initial state deviations. Furthermore, different number of vehicles are placed on the two roads to verify the traffic density balancing property of the upper-level controller. To avoid overlaps between trajectories and provide a clear view of the trajectory evolution, a total number of 5 CAVs are generated. Note that, the number of vehicles does not impact the results since the upper-level controller is designed without explicit dependence on the number of vehicles and the asymptotic stability is proved for each vehicle. The MIP controller determines each CAV’s merging sequence and convey it to the CAV’s lower-level longitudinal controller to initiate the merging process. We compare the position and state evolution, along with quantitative criteria of the MIP controller’s merging process to a FIFO baseline, highlighting both the local stability characteristic and the proposed MIP controller’s superiority.

\begin{figure}[h]
    \centering
    \setlength{\abovecaptionskip}{0pt}
    \subfigure[]{\includegraphics[width=0.2\textwidth]{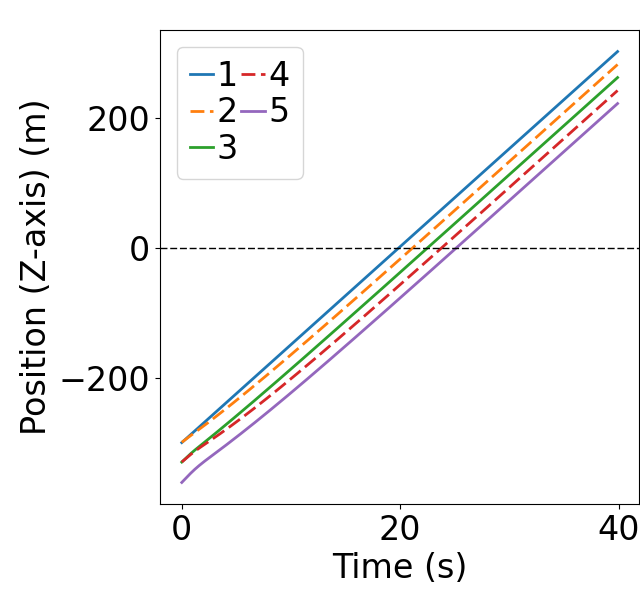}}
    \subfigure[]{\includegraphics[width=0.2\textwidth]{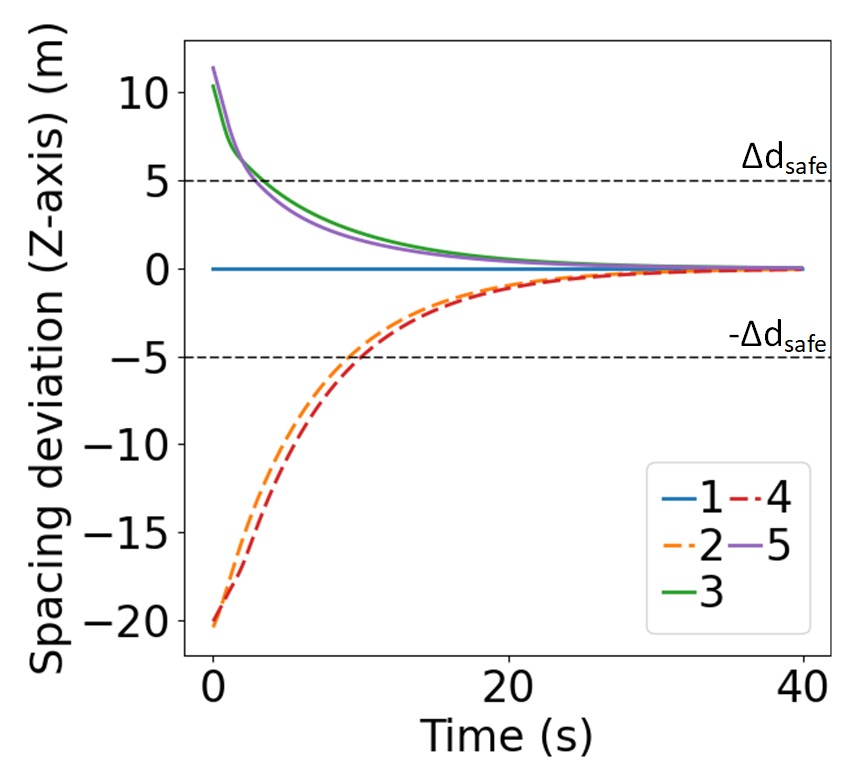}}
    \subfigure[]{\includegraphics[width=0.2\textwidth]{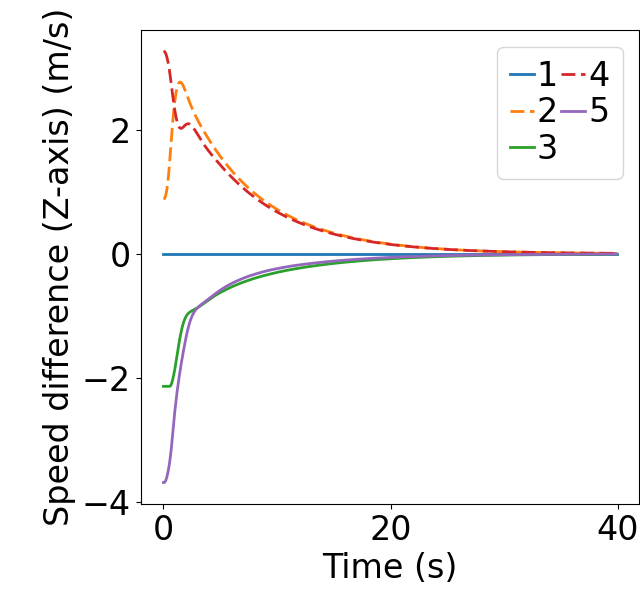}}
    \subfigure[]{\includegraphics[width=0.2\textwidth]{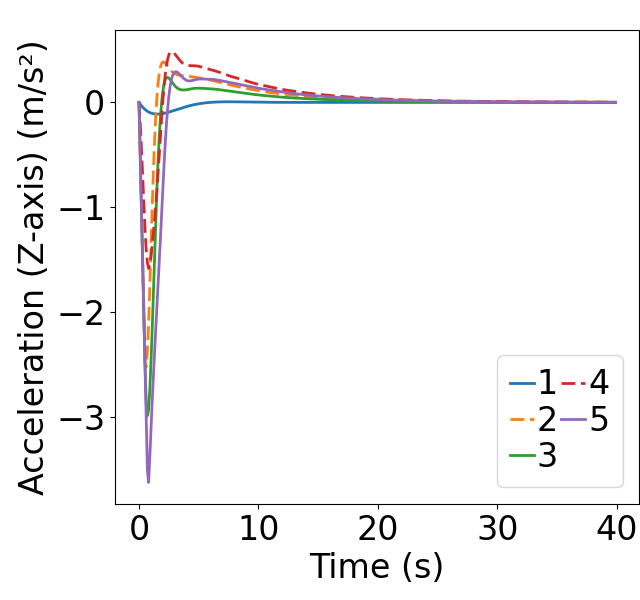}}
    \caption{CAV position and controller state evolution of experiment in scenario 1 with MIP. (a) Position evolution. (b) Spacing deviation evolution. (c) Speed difference evolution. (d) Acceleration evolution.}
    \label{fig:mip}
\end{figure}

\begin{figure}[h]
    \centering
    \subfigure[]{\includegraphics[width=0.2\textwidth]{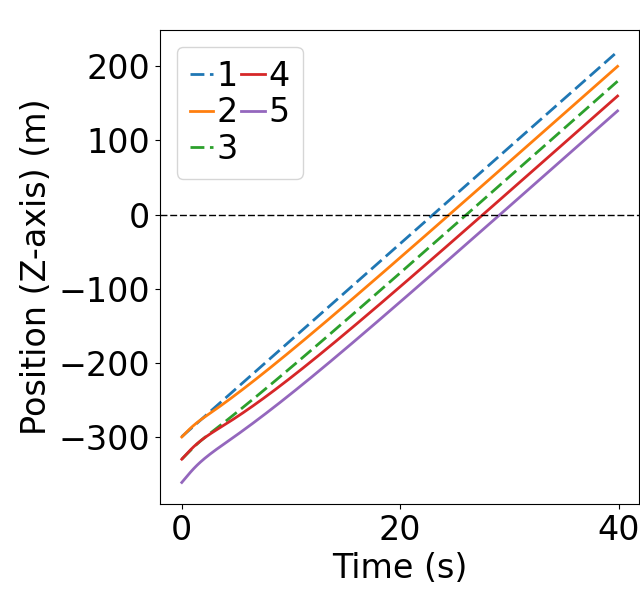}}
    \subfigure[]{\includegraphics[width=0.2\textwidth]{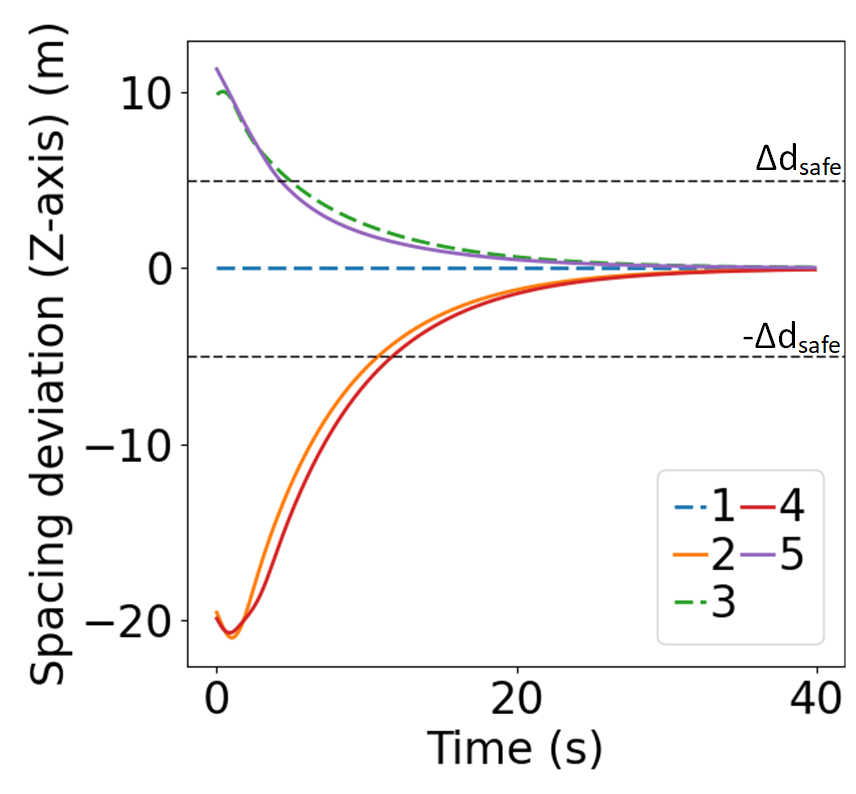}}
    \subfigure[]{\includegraphics[width=0.2\textwidth]{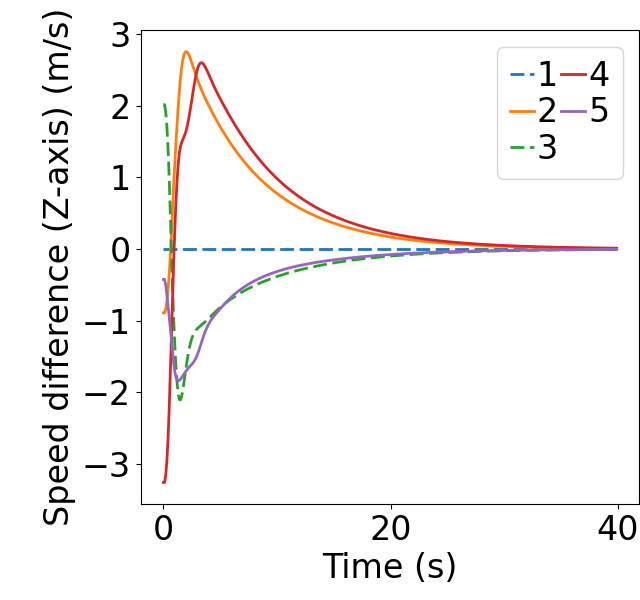}}
    \subfigure[]{\includegraphics[width=0.2\textwidth]{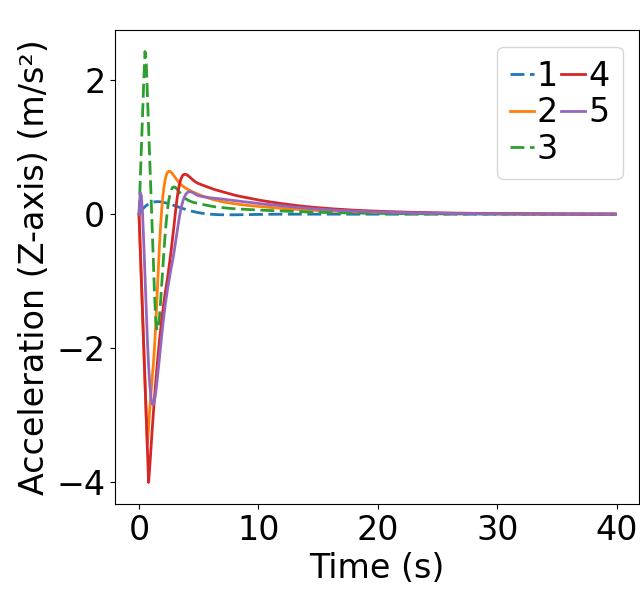}}
    \caption{CAV position and controller state evolution of experiment in scenario 1 with FIFO. (a) Position evolution. (b) Spacing deviation evolution. (c) Speed difference evolution. (d) Acceleration evolution.}
    \label{fig:fifo}
\end{figure}

In scenario 1, we simulate 3 mainline CAVs and 2 on-ramp CAVs. The leading CAVs on both lanes start at -300m relative to the virtual Z-axis, with a uniformly distributed deviation from -1$m$ to 1$m$. Successive CAVs are placed at 30$m$ intervals on both lanes, each with deviations ranging from -1$m$ to 1$m$. Initial velocities of the CAVs incrementally increase on the mainline and decrease on the on-ramp, starting from 15$m/s$, with uniformly distributed deviations between -0.8$m/s$ and 0.8$m/s$. Initial accelerations are randomly set between -0.5$m/s^2$ and 0.5$m/s^2$. The weights are selected as: $Q_{lon}=diag(0.01,0.02,0.01)$, $R_{lon}=0.01$, by setting $\varepsilon=0.5$, based on Proposition \ref{prop asymptotic 1}, we obtain the corresponding $\beta=1600$ to ensure asymptotic local stability.

Fig. \ref{fig:mip} displays the evolution of position and states during the merging process using MIP. The dashed lines represent on-ramp CAVs, while the solid lines represent mainline CAVs. \ref{fig:mip}(a) demonstrates that MIP indeed prioritizes merging for the mainline, which has higher traffic density. As shown in Fig. \ref{fig:mip}(b-c), the MIP controller takes into account both spacing deviation and speed difference from the predecessor, allowing the initial speed difference to directly counteract the spacing deviation. Moreover, as shown in Fig. \ref{fig:mip}(d), the acceleration evolution using MIP controller exhibits relatively small peaks. In contrast, as shown in Fig. \ref{fig:fifo}(a-c), the FIFO strategy solely considers the position of the CAVs, it also results in a situation where CAV 2, 3, and 4 have to accelerate/decelerate until the speed differences evolve to the opposite direction before the spacing deviation begins to be eliminated. This results in an initial surge in the spacing deviation. Furthermore, the acceleration evolution shows relatively large peaks.

\begin{figure}[h]
    \centering
    \subfigure[]
    {\includegraphics[width=0.2\textwidth]{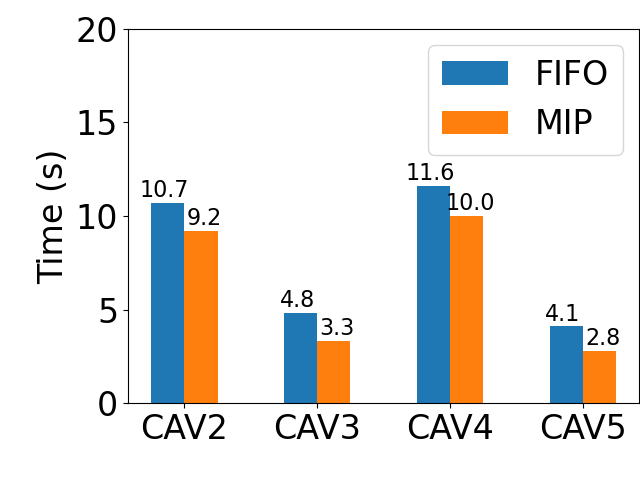}}
    \subfigure[]{\includegraphics[width=0.2\textwidth]{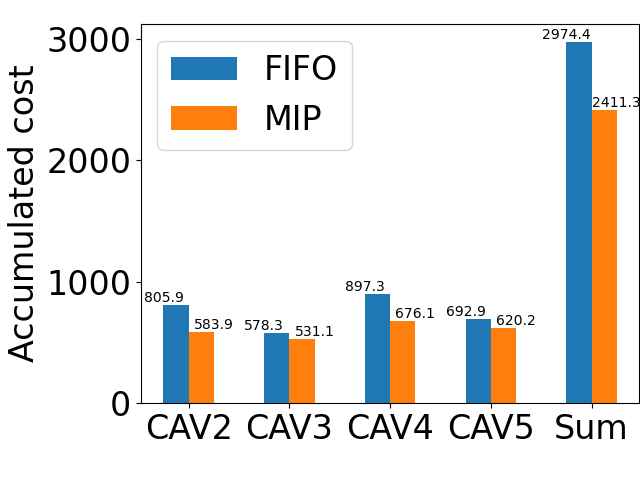}}
    \caption{Quantitative comparison of experiments in scenario 1. (a) Time spent to converge within [$-\Delta d_{safe},~{}\Delta d_{safe}\text{]}$. (b) Accumulated cost before converging.}
    \label{fig:veh_move}
\end{figure}
Quantitative comparison results of the merging process are shown in Fig. \ref{fig:veh_move}. It is observed that the MIP controller leads to less time spent for $\Delta d_i(t)$ to converge within a safe range [$-\Delta d_{safe},~{}\Delta d_{safe}$] and smaller accumulated cost before converging for the lower-level longitudinal controller, wherein the cost  at each time step is given by (\ref{stage cost}).

Both Fig.\ref{fig:mip} and Fig.\ref{fig:fifo} showcase that the states of all CAVs are regulated to zero over time, indicating local asymptotic stable characteristic of the lower-level longitudinal controller. Moreover, both Fig. \ref{fig:mip}(b) and Fig. \ref{fig:fifo}(b) demonstrate that the spacing deviation for each CAV converges to the predefined safety range $[-\Delta d_{safe}, \Delta d_{safe}]$ well before reaching the merging point, underscoring the safety of the merging process.

The computational time of the longitudinal controllers of the first experiment are summarized in Table \ref{tab:local}, detailing the upper-level controller computational time, and the average and maximum computational time of the lower-level longitudinal controller. As described in Section \ref{sec2}, the upper-level controller operates only when a new vehicle enters the control area, ensuring minimal computational demand.  Even if two vehicles enter the control area in quick succession, the system can delay sequence assignments until after the new output is calculated, as sequence assignment does not necessitate immediate real-time response. The lower-level controller, meanwhile, reports both average and maximum computational time well below the discrete time interval $T_s=0.1s$, this allows for new control inputs to be computed concurrently with the execution of the previous one. Therefore, the longitudinal controllers meet the real-time computational requirements.

\begin{table}[h]
\centering
\tiny
\setlength{\abovecaptionskip}{0pt}
\caption{Computational Time of Longitudinal Controllers in Scenario 1 Experiment}
\label{tab:local}
 
    \centering
    \begin{adjustbox}{width=0.9\linewidth}
    \begin{tabular}[t]{|>{\centering\arraybackslash}m{1cm}|>{\centering\arraybackslash}m{1.1cm}|>{\centering\arraybackslash}m{1.3cm}|}
        \hline
        \textbf{\textit{Upper}} & \textbf{\textit{Lower average}} & \textbf{\textit{Lower maximum}} \\
        
        \hline
        $0.086s$ & $0.047s$  & $0.059s$\\
        \hline
    \end{tabular}
    \end{adjustbox}
\end{table}


\subsection{String Stability Experiments\label{sec6.2}}
In this subsection, we examine the $l_2$-norm string stability of the lower-level longitudinal controller with scenario 2.

\begin{figure}[h]
    \centering
    \setlength{\abovecaptionskip}{0pt}
    \subfigure[]{\includegraphics[width=0.2\textwidth]{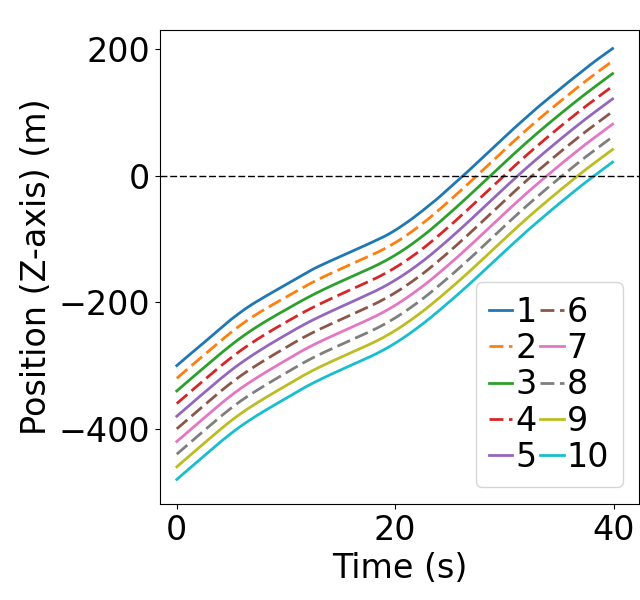}}
    \subfigure[]{\includegraphics[width=0.2\textwidth]{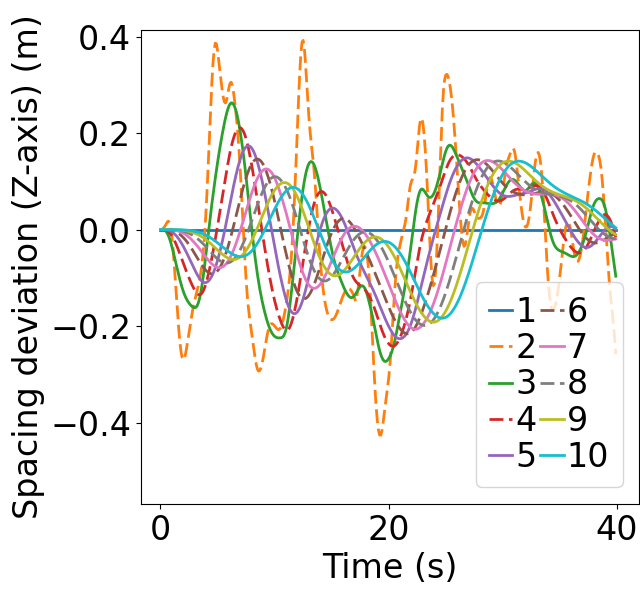}}
    \subfigure[]{\includegraphics[width=0.215\textwidth]{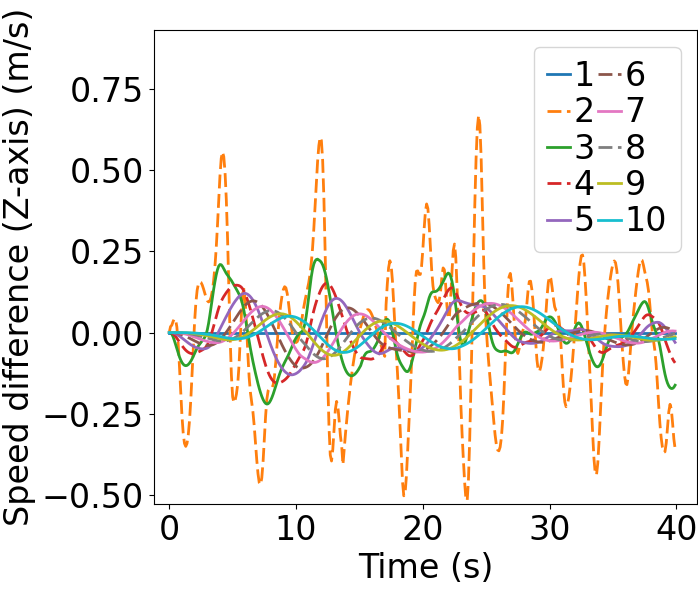}}
    \subfigure[]{\includegraphics[width=0.2\textwidth]{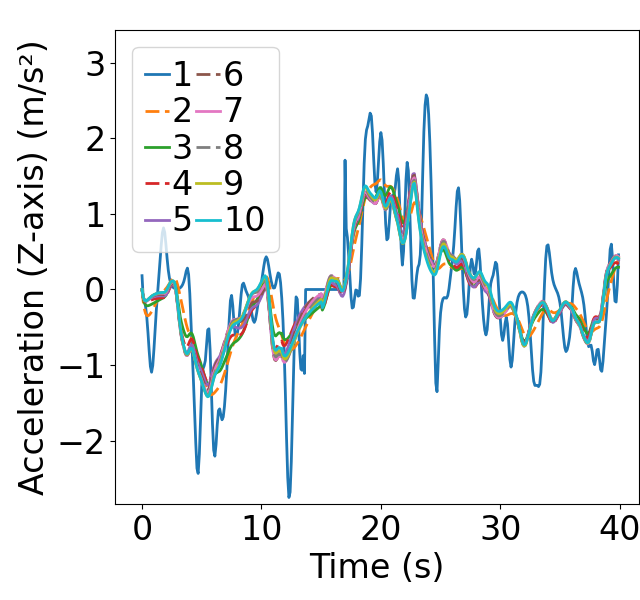}}
    \caption{CAV position and controller state evolution of the scenario 2 experiment. (a) Position evolution. (b) Spacing deviation evolution. (c) Speed difference evolution. (d) Acceleration evolution.}
    \label{fig:string_stable}
\end{figure}


It is important to clearly show the disturbance damping along the vehicular string, hence a scenario is designed with disturbances continuously applied to the leading vehicle and the impact of initial deviations eliminated. To better present the disturbance damping, a total number of 10 CAVs are generated. Note that, the number of vehicles does not impact the results since string stability is proved for each vehicle. In this setup, 6 mainline CAVs are positioned at intervals of 20 meters, complemented by 4 on-ramp CAVs spaced every 40 meters. Each CAV starts with an initial speed of 15 meters per second. Notably, CAV1's acceleration profile is derived from the real-world NGSIM dataset \cite{NGSIM}, while the remaining CAVs commence with zero initial acceleration. In this configuration, CAVs 2-10 start without any initial disturbances, but the leading CAV1 consistently introduces disturbances throughout the merging process. In this scenario, the weights are selected the same as the ones in subsection \ref{sec6.1}, the corresponding linear feedback and feedforward gains $K_b=[0.1849$, $10.5855$, $-4.9804]$, $k_f=5.8356$ satisfy the condition (\ref{string stable consition2}), indicating that the platoon should display $l_2$-norm string stability.

Fig. \ref{fig:string_stable} presents the simulation results of scenario 2.  The figure depicts the evolution of position and states of the 10 CAVs. The dashed lines represent on-ramp CAVs while the solid lines represent mainline CAVs. It is observed that the disturbances are dampened along CAV string in all three states, showcasing the $l_2$-norm string stable characteristic. Furthermore, Fig. \ref{fig:string_stable}(b) illustrates that the spacing deviation for each CAV starts within the predefined safety range $[-\Delta d_{safe}, \Delta d_{safe}]$, and stays within the range for the whole merging process, underscoring the safety of the merging process.

Similar to Table \ref{tab:local}, Table \ref{tab:string} summarizes the computational time of the longitudinal controllers. Following the same rationale in Subsection \ref{sec6.1}, it is again demonstrated that the longitudinal controllers satisfy the necessary real-time computational requirements.

\begin{table}[h]
\centering
\tiny
\setlength{\abovecaptionskip}{0pt}
\caption{Computational Time of Longitudinal Controllers in Scenario 2 Experiment}
\label{tab:string}

    \centering
    \begin{adjustbox}{width=0.9\linewidth}
    \begin{tabular}[t]{|>{\centering\arraybackslash}m{1cm}|>{\centering\arraybackslash}m{1.1cm}|>{\centering\arraybackslash}m{1.3cm}|}
        \hline
        \textbf{\textit{Upper}} & \textbf{\textit{Lower average}} & \textbf{\textit{Lower maximum}} \\
        
        \hline
        $0.101s$ & $0.045s$  & $0.061s$\\
        \hline
    \end{tabular}
    \end{adjustbox}
\end{table}

\subsection{Lower-level Lateral Controller Experiments\label{sec6.3}}

The linear MPC-based lower-level lateral controller regulates lateral and heading angle deviations, allowing the framework to handle curvy ramps and keep CAVs centered in their lanes. As the lateral controller operates locally on each CAV and is not dependent on the movement of other CAVs, we examine it on a single on-ramp CAV in a two dimensional scenario where it tracks the centerline of the ramp. Lateral and heading angle deviations primarily occur at the initial time points and at moments when disturbances (e.g., instantaneous changes in curvature) are introduced. To provide a clearer view of these critical time points, the CAV is initially positioned closer to the merging point.

\begin{figure}[h]
    \centering
    \setlength{\abovecaptionskip}{0pt}
    {\includegraphics[width=0.49\textwidth]{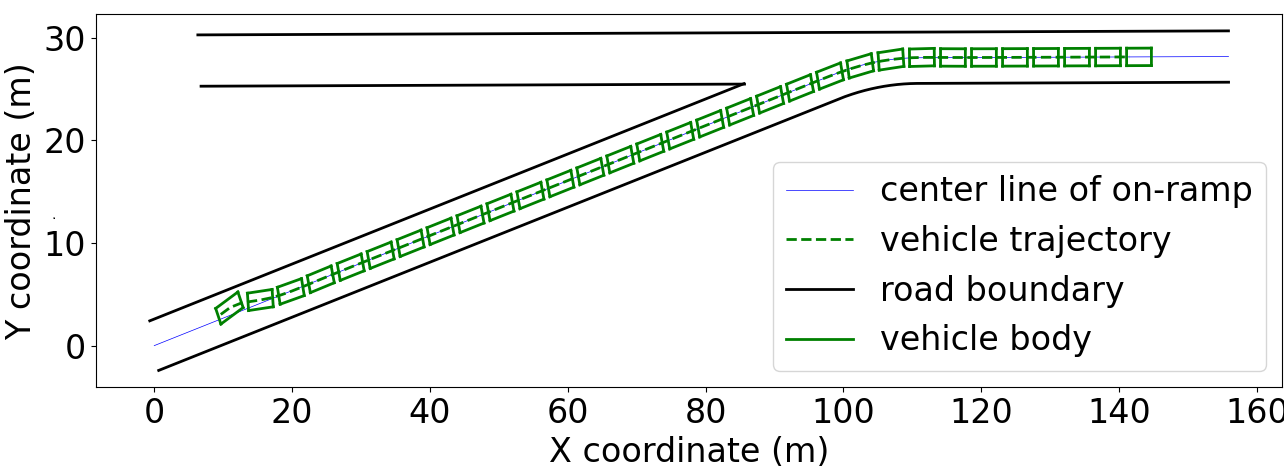}}
    \caption{Road geometry and CAV trajectory of scenario 3 experiment}
    \label{fig:lat_scenario}
\end{figure}
\begin{figure}[h]
    \centering
    \setlength{\abovecaptionskip}{0pt}
    \subfigure[]{\includegraphics[width=0.2\textwidth]{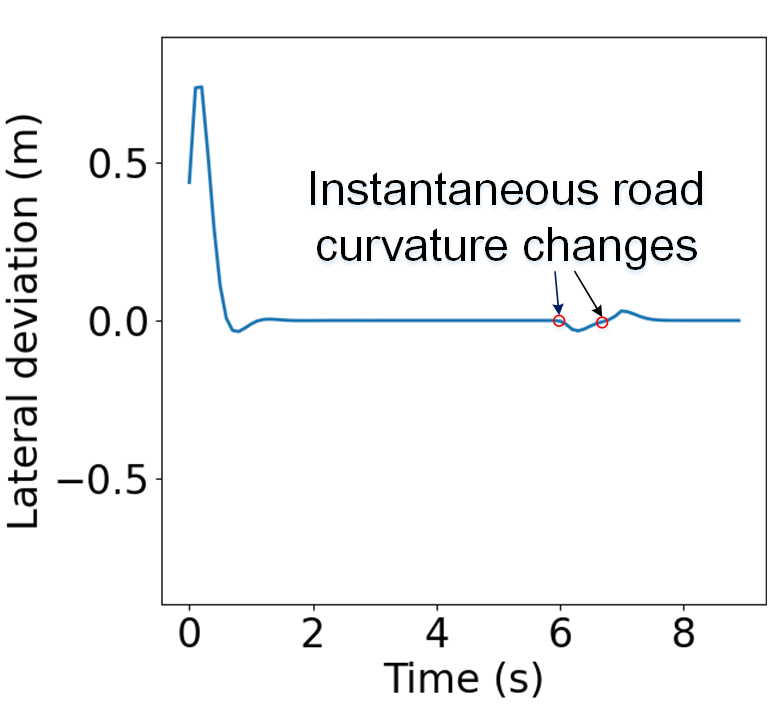}}
    \subfigure[]{\includegraphics[width=0.203\textwidth]{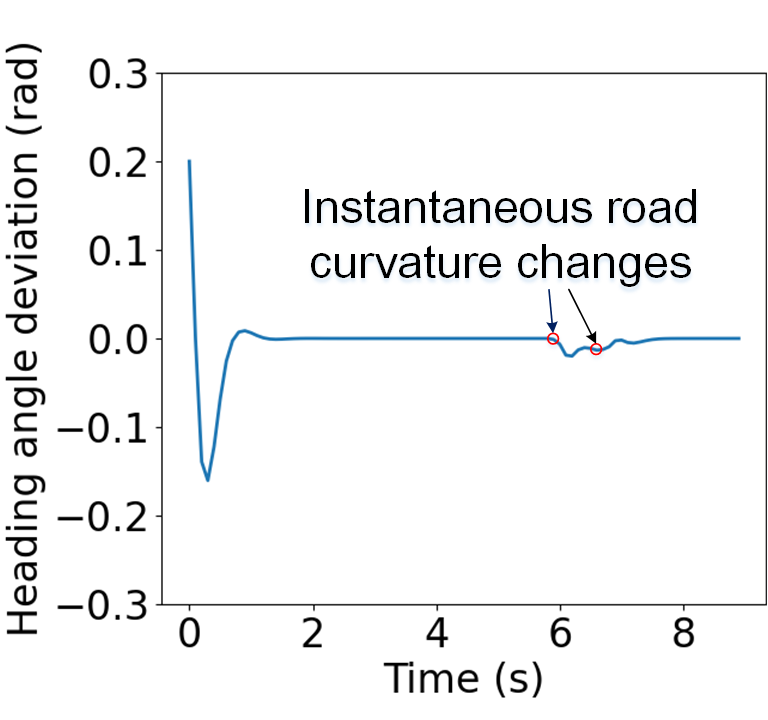}}
    \subfigure[]
    {\includegraphics[width=0.2\textwidth]{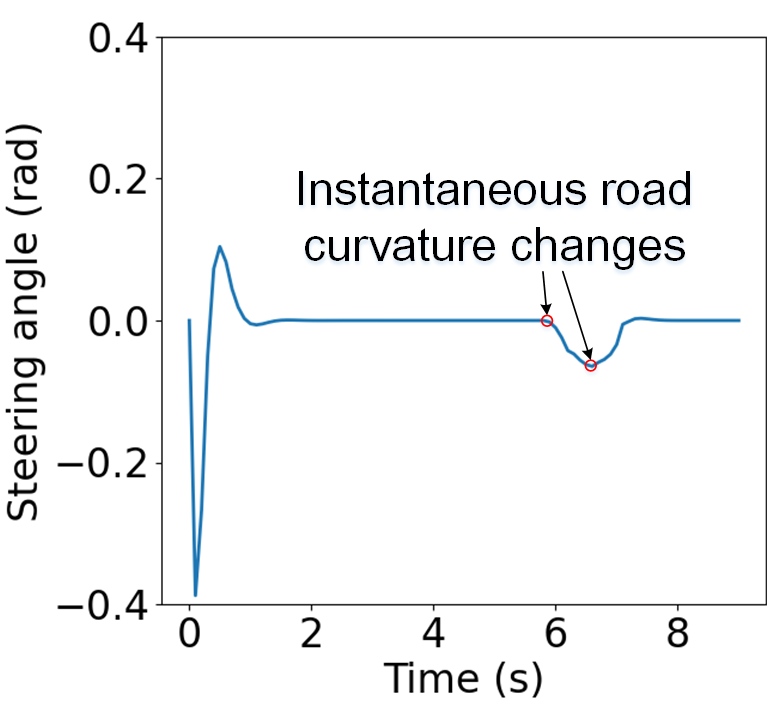}}
    \caption{Tracking deviation and control input evolution of scenario 3 experiment. (a) Lateral deviation. (b) Heading angle deviation. (c) Steering angle.}
    \label{fig:lat_result}
\end{figure}

In scenario 3, an on-ramp CAV starts at -110$m$ along the virtual Z-axis, with an initial velocity of 15$m/s$ and zero acceleration. It trails a mainline CAV positioned initially at -90$m$ on the Z-axis, also moving at 15$m/s$. The mainline CAV's acceleration profile is extracted from the NGSIM dataset, similar to Scenario 2. Post each time step of the lower-level longitudinal controller, the lower-level lateral controller is activated, leveraging the longitudinal controller's velocity prediction for enhanced CAV state forecasting. The on-ramp CAV begins with a lateral deviation and heading angle deviation of 0.42$m$ and 0.2$rad$, respectively. 

For the purpose of visualization, the trajectory of the on-ramp CAV is illustrated in Fig. \ref{fig:lat_scenario}, together with part of the road geometry. The vehicle body is plotted every 3 time steps to better visualize the CAV's movement. The evolution of the lateral deviation, heading angle deviation, and steering angle is shown in Fig. \ref{fig:lat_result}. It is observed that the initial lateral deviation and heading angle deviation are quickly regulated, and the steering angle remains relatively smooth with fast convergence. The lateral deviation, heading angle deviation, and steering angle all exhibit two disturbances during the merging process, caused by the instantaneous changes in road curvature. However, the disturbances are rapidly mitigated.

The average and maximum computational time of the lateral controller are
presented in Table \ref{tab:lateral}. It is observed that both the average and maximum computational time are well below the discrete time interval $T_s=0.1s$, facilitating concurrent computation of a new control inputs alongside the execution of the previous one. Consequently, the lateral controller meets the required real-time computational standards.

\begin{table}[h]
\centering
\tiny
\setlength{\abovecaptionskip}{0pt}
\caption{Computational Time of the Lateral Controller in Scenario 3 Experiment}
\label{tab:lateral}

    \centering
    \begin{adjustbox}{width=0.6\linewidth}
    \begin{tabular}[t]{|>{\centering\arraybackslash}m{1.1cm}|>{\centering\arraybackslash}m{1.3cm}|}
        \hline \textbf{\textit{Average}} & \textbf{\textit{Maximum}} \\
        \hline
         $0.036s$  & $0.053s$\\
        \hline
    \end{tabular}
    \end{adjustbox}
\end{table}

\section{CONCLUSIONS\label{sec7}}

This paper deals with a CAVs on-ramp merging problem with a hierarchical control framework. The control is designed with two levels. At the upper level, it employs mixed-integer linear programming with a multi-scale cost function to generate merging sequences. This takes into account both the microscopic dynamics, like relative positions and velocities of vehicle pairs, and the macroscopic traffic distribution on mainline and on-ramps. The lower level introduces a distributed MPC-based longitudinal controller, focused on minimizing spacing deviation, speed difference, acceleration, and jerk for each CAV, while incorporating a safety stage cost and terminal constraint for enhanced safety. Asymptotic local stability and $l_2$ norm string stability are proved with mathematical derivations, and the longitudinal controller’s initial feasibility is significantly enhanced. Moreover, a linear-MPC based lateral controller minimizes lateral and heading angle deviations, and steering angles of each CAV, making the framework applicable to two-dimensional scenarios.

Through multi-scenario simulations and comparisons, the proposed framework greatly improves traffic efficiency and cost-effectiveness compared to a FIFO method. The whole control framework exhibits stability characteristics at both vehicle and system levels by the fact that the lateral deviations converge rapidly, as well as ensured asymptotic local stability and string stability longitudinally. Furthermore, it is shown that the control framework satisfies real-time computational requirements. Future work involves developing a lane changing maneuver for scenarios including multiple mainline lanes.

\bibliographystyle{IEEEtran}
\bibliography{root}
\begin{IEEEbiography}[{\includegraphics
[width=1in,height=1.25in,clip,
keepaspectratio]{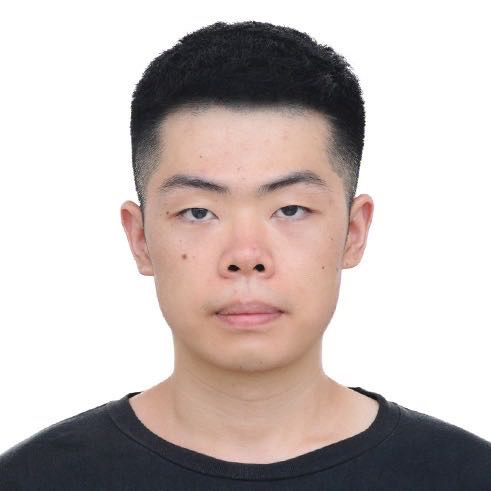}}]
{Sixu Li}
received the B.S. degree in Engineering
Mechanics from Hunan University, Changsha, China, in
2020, and the master's degree in
Mechanical Engineering from UC Berkeley, Berkeley, CA, USA, in 2022. He is currently pursuing the Ph.D. degree in Interdisciplinary Engineering with Texas A\&M University, College Station, TX, USA. 

His current research interests include dynamics and
control, MPC, optimization, and reinforcement
learning for autonomous driving, intelligent transportation systems, and robotics.
\end{IEEEbiography}

\begin{IEEEbiography}[{\includegraphics
[width=1in,height=1.25in,clip,
keepaspectratio]{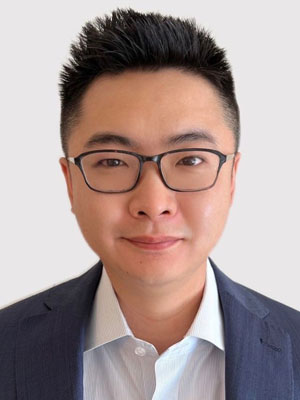}}]
{Yang Zhou}
(Member, IEEE) received the master’s
degree in Civil and Environmental Engineering from
the University of Illinois at Urbana-Champaign,
Champaign, IL, USA, in 2015, and the Ph.D. degree
in Civil and Environmental Engineering from the University of Wisconsin–Madison, WI, USA, in 2019. He worked as a PostDoc Researcher (2019-2022) at the Department of Civil and Environmental Engineering, University of Wisconsin-Madison
Madison. He is currently an Assistant Professor at the Zachry Department of Civil \& Environmental Engineering, Texas A\&M University, College Station, TX, USA. 

His main research directions are connected automated vehicles robust control, interconnected system stability analysis, traffic big data analysis, and microscopic traffic flow modeling.
\end{IEEEbiography}

\begin{IEEEbiography}[{\includegraphics
[width=1in,height=1.25in,clip,
keepaspectratio]{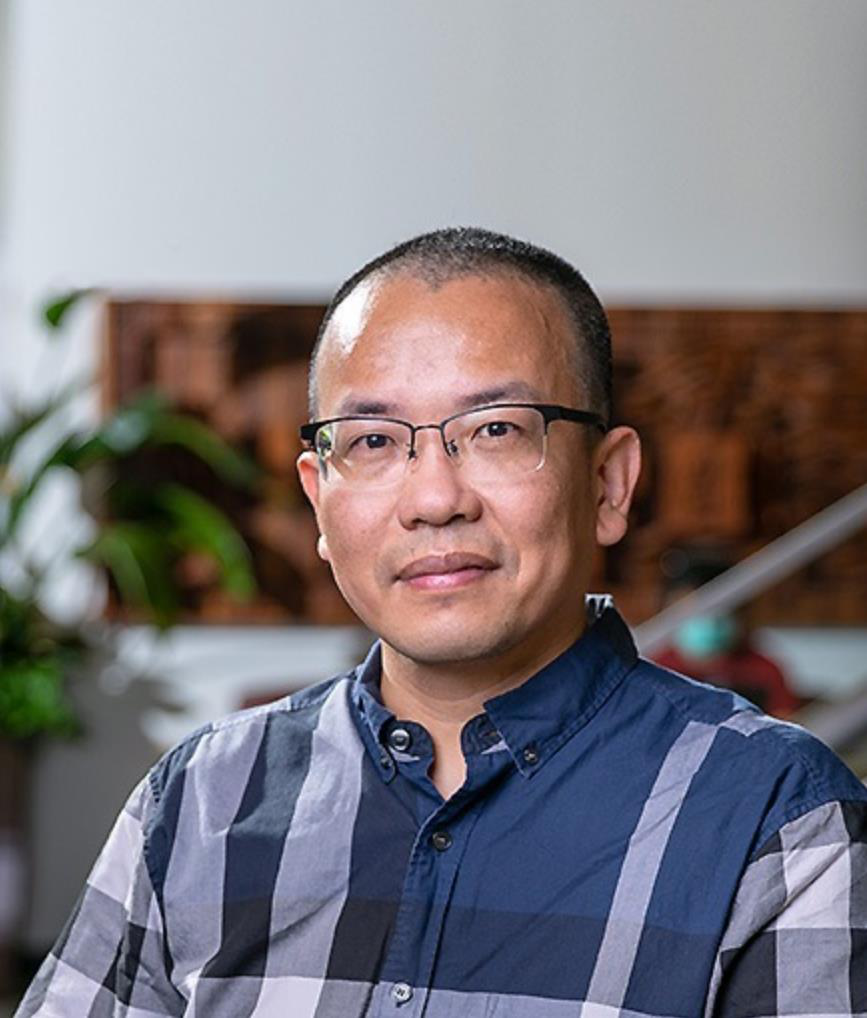}}]
{Xinyue Ye}
received a PhD degree in Geography from University of California at Santa Barbara 2010. He worked as an Assistant Professor (2010-2013) at Bowling Green State University; Assistant/Associate Professor (2013-2018) at Kent State University; Associate Professor (2018-2020) at New Jersey Institute of Technology. He is now Full Professor at the Department of Landscape Architecture \& Urban Planning at Texas A\&M University and holds Harold L. Adams Endowed Professorship with joint appointments across multiple colleges. He directs the GEOSAT Center, established by the Texas A\&M Board of Regents. 

His main research interests are human dynamics and urban artificial intelligence, particularly in the context of convergence research. He is an Editor of Computational Urban Science.
\end{IEEEbiography}

\begin{IEEEbiography}[{\includegraphics
[width=1in,height=1.25in,clip,
keepaspectratio]{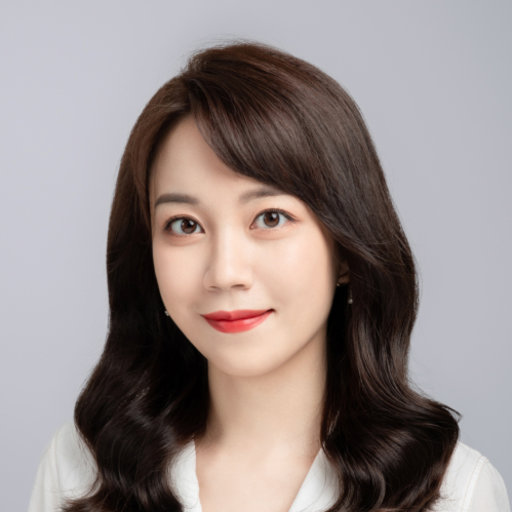}}]
{Jiwan Jiang}
received the B.S. degree from Southwest Jiaotong University, Chengdu, China, in 2018, and the master’s degree from the School of Transportation, Southeast University, Nanjing, China, in 2021. She is currently pursuing the Ph.D. degree
in Civil and Environmental Engineering with the University of Wisconsin–Madison, WI, USA.

Her main research interests focus on connected automated vehicle control and operation, harnessing stochastic methods for car-following model calibration, and parallel computing.
\end{IEEEbiography}

\begin{IEEEbiography}[{\includegraphics
[width=1in,height=1.25in,clip,
keepaspectratio]{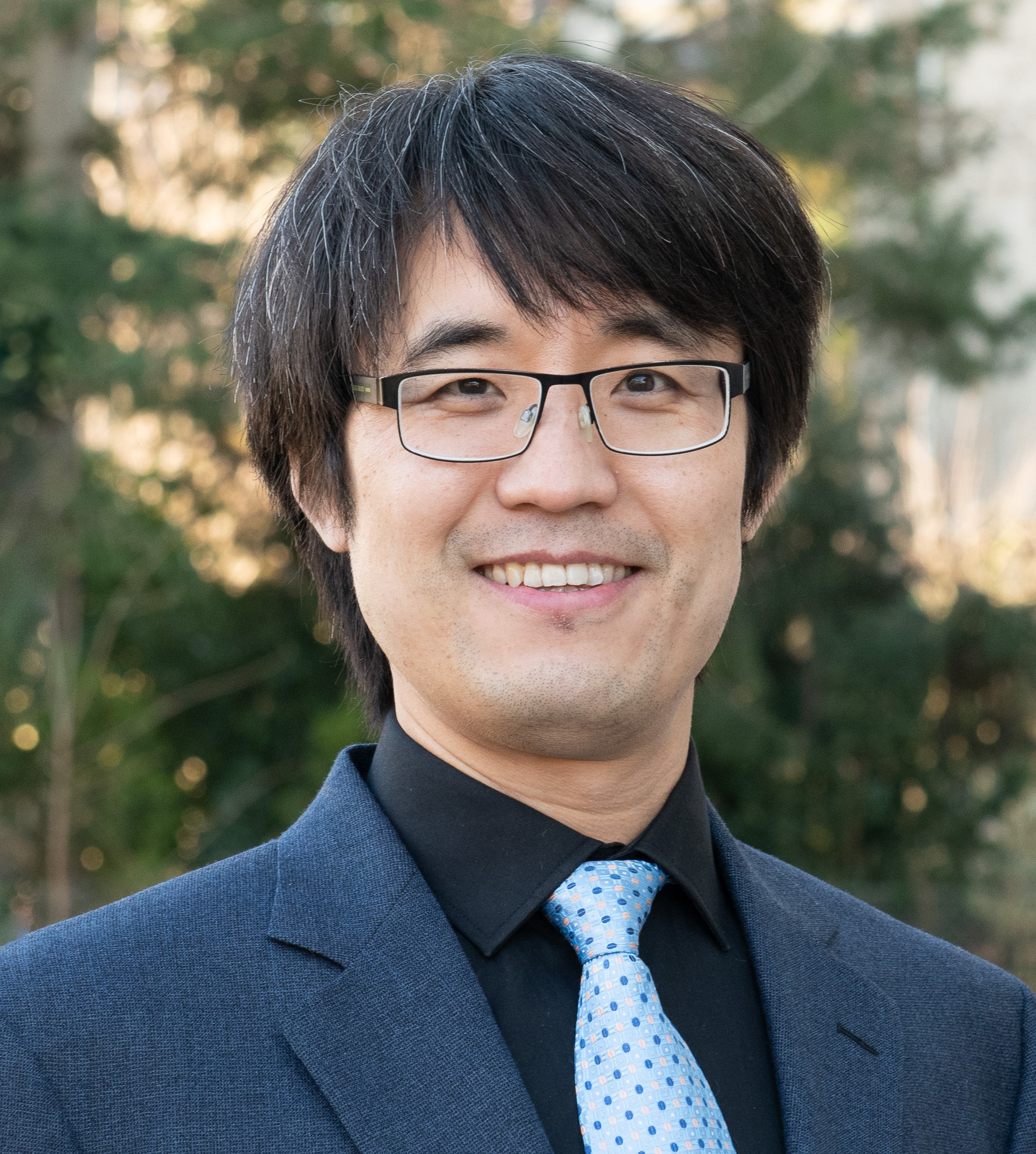}}]
{Meng Wang}
(Member, IEEE) received an M.Sc. degree from Research Institute of Highway and a PhD degree (Hons.) from TU Delft, in 2006 and 2014, respectively. He worked as a PostDoc Researcher (2014-2015) at the Faculty of Mechanical Engineering, TU Delft, and as an Assistant Professor (2015-2021, tenured since 2019) with the Department of Transport and Planning. Since 2021, he has been a Full Professor and Head of the Chair of Traffic Process Automation, “Friedrich List” Faculty of Transport and Traffic Sciences, TU Dresden. 

His main research interests are control design and impact assessment of Cooperative Intelligent Transportation Systems. He is an Associate Editor of IEEE Transactions on Intelligent Transportation Systems and Transportmetrica B.
\end{IEEEbiography}

\end{document}